\documentclass[11pt,reqno]{amsart}
\usepackage{amssymb}
\newtheorem{theorem}{Theorem}[section]

\newtheorem{remark}{Remark}[section]

\newfont{\bb}{msbm10 at 12pt}

\newcommand{\mysection}[1]{\section{#1}\setcounter{equation}{0}}

\newcommand{\bal}{\begin{align}}     \newcommand{\eal}{\end{align}}
\newcommand{\ba}{\begin{array}}      \newcommand{\ea}{\end{array}}
\newcommand{\bc}{\begin{center}}     \newcommand{\ec}{\end{center}}
\newcommand{\be}{\begin{enumerate}}  \newcommand{\ee}{\end{enumerate}}
\newcommand{\beq}{\begin{eqnarray}}  \newcommand{\eeq}{\end{eqnarray}}
\newcommand{\beQ}{\begin{eqnarray*}} \newcommand{\eeQ}{\end{eqnarray*}}
\newcommand{\bi}{\begin{itemize}}    \newcommand{\ei}{\end{itemize}}
\newcommand{\bt}{\begin{tabular}}    \newcommand{\et}{\end{tabular}}
\newcommand{\bdm}{\begin{displaymath}} \newcommand{\edm}{\end{displaymath}}

\newcommand{\ls}{\setlength{\baselineskip}{12pt}
                 \setlength{\parskip}{3mm}}

\usepackage{upgreek}
\usepackage{tensor}
\usepackage{mathrsfs}

\usepackage{enumerate}

\begin{document}


\allowdisplaybreaks

\title[Peeling property]
{Peeling property for the Einstein scalar field equations with the nonzero cosmological constant}

\author[J Li]{Jialue Li$^{1}$}
\author[X Zhang]{Xiao Zhang$^{2,3,4}$}

\address[]{$^{1}$School of Mathematical Sciences, Peking University, Beijing 100871, People's Republic of China}
\address[]{$^{2}$State Key Laboratory of Mathematical Sciences, Academy of Mathematics and Systems Science, Chinese Academy of Sciences, Beijing 100190, People's Republic of China}
\address[]{$^{3}$School of Mathematical Sciences, University of Chinese Academy of Sciences, Beijing 100049, People's Republic of China}
\address[]{$^{4}$Guangxi Center for Mathematical Research, Guangxi University, Nanning 530004, Guangxi, People's Republic of China}

\email{lijialue@pku.edu.cn$^{1}$}
\email{xzhang@amss.ac.cn$^{2,3,4}$}

\begin{abstract}
Inspired by interaction of gravitational waves, dark matters and dark energy, we study the Einstein scalar field equations with the nonzero cosmological constant for the Bondi-Sachs metrics. We provide the asymptotic expansions under the outgoing radiation condition. Unlike the case of the zero cosmological constant, the asymptotic expansions depend on three additional $\Lambda$-independent functions $\tilde{X}(u, \theta, \phi)$, $\tilde{Y}(u, \theta, \phi)$ and $\tilde{I}(u)$. We prove the peeling property for the Einstein scalar field equations if $\tilde{I}=0$, and derive a loss formula of the Bondi energy-momentum for the nonzero cosmological constant. We remark that, for certain real data from gravitational waves, the loss formula is likely to indicate the loss property of the Bondi energy.
\end{abstract}
\keywords{Bondi-Sachs metric; Scalar field; Peeling property; Bondi energy-momentum.}

\subjclass[2000]{53C50, 83C35}

\date{}

\maketitle \pagenumbering{arabic}
\pagenumbering{arabic}

\mysection{Introduction}
\ls

In general relativity, a spacetime is a four-dimensional Lorentzian manifold whose metric $\mathbf{g}$ satisfies the Einstein field equations
\begin{equation}\label{eq:EinsteinTotal}
    R_{\mu\nu}-\frac{R}{2}\mathbf{g}_{\mu\nu}+\Lambda \mathbf{g}_{\mu\nu}=T_{\mu\nu},
\end{equation}
where $R_{\mu\nu}$ is the Ricci curvature, $R$ is the scalar curvature, $\Lambda$ is the cosmological constant and $T_{\mu\nu}$ is the energy-momentum tensor of matter. In particular, it is a vacuum when $T_{\mu\nu}$ vanishes. 

Gravitational waves are vacuum, wave-like spacetime metrics radiating energy. When the cosmological constant is zero, gravitational waves were firstly studied by Bondi, van der Burg, Metzner for axially symmetric isolated spacetimes and by Sachs for asymptotically flat spacetimes \cite{BBM, S}. They introduced the following Bondi-Sachs metrics 
\begin{equation}
\begin{aligned}\label{eq:BondiSachsMetric}
    \mathbf{g} =&-\bigg( \frac{\mathrm{e}^{2\beta}V}{r} -r^2\Big(\mathrm{e}^{2\gamma}U^2\cosh(2\delta)+2 UW\sinh(2\delta) \\
              &+\mathrm{e}^{-2\gamma}W^2\cosh(2\delta)\Big)\bigg)\mathrm{d}u^2-2\mathrm{e}^{2\beta}\mathrm{d}u\mathrm{d}r \\
              & -2r^2\bigg(\mathrm{e}^{2\gamma}U\cosh(2\delta)+W\sinh(2\delta)\bigg)\mathrm{d}u\mathrm{d}\theta \\
	          & -2r^2\bigg(\mathrm{e}^{-2\gamma}W\cosh(2\delta)+U\sinh(2\delta)\bigg)\sin\theta \mathrm{d}u\mathrm{d}\phi \\
	          & +r^2\bigg(\mathrm{e}^{2\gamma}\cosh(2\delta)\mathrm{d}\theta^2+2\sinh(2\delta)\sin\theta\mathrm{d}\theta \mathrm{d}\phi \\
              & +\mathrm{e}^{-2\gamma}\cosh(2\delta)\sin^2\theta\mathrm{d}\phi^2\bigg)
\end{aligned}
\end{equation}
in coordinates $\{u,\,r,\,\theta,\,\phi\}$ ($u$ is retarded time)
\begin{equation*}
    -\infty<u<\infty,\quad r\geqslant 0,\quad 0\leqslant \theta\leqslant \pi,\quad
    0\leqslant \phi\leqslant 2\pi,
\end{equation*}
where $\beta$, $\gamma$, $\delta$, $U$, $V$ and $W$ are smooth functions of $u$, $r$, $\theta$ and $\phi$, which are defined on $S^2$ for each $u$, i.e., they and their derivatives take the same value at $\phi=0, \,2\pi$. Assuming the {\it outgoing radiation condition} 
\begin{align}
\gamma = \frac{c(u, \theta, \phi)}{r} +O\left(\frac{1}{r^3}\right), \quad \delta =\frac{d(u, \theta, \phi)}{r} +O\left(\frac{1}{r^3}\right), \label{outgoing}
\end{align}
the vacuum Einstein field equations imply \cite{BBM, S, vdB}
\begin{align*}
\beta = O\left(\frac{1}{r^2}\right), \,  U=O\left(\frac{1}{r^2}\right), \, W=O\left(\frac{1}{r^2}\right), \, V=r-2M(u, \theta, \phi)+O\left(\frac{1}{r}\right).
\end{align*}
Physically, $c_u$, $d_u$ are referred to as {\it news functions}, and $M$ is referred to as {\it mass aspect}. Denote
\begin{equation*}
    n^0=1,\quad n^1=\sin\theta\cos\phi,\quad n^2=\sin\theta\sin\phi,\quad n^3=\cos\theta.
\end{equation*}
The Bondi energy-momentum of each null hypersurface is defined as \cite{BBM}
\begin{align*}
    m_\nu(u)=\frac{1}{4\pi}\int_{S^2}M(u,\theta,\phi)n^\nu\mathrm{d}S, \quad \nu=0, 1, 2, 3.
\end{align*}
If the following conditions hold,
\begin{align}
    \int_0 ^{2\pi} c(u, 0, \phi)\mathrm{d}\phi=0, \quad \int_0 ^{2\pi} c(u, \pi, \phi)\mathrm{d}\phi=0, \label{hyz-u}
\end{align}
then the Bondi energy-momentum satisfies the loss property of energy-momentum for any $u$ \cite{BBM, Z2, HYZ}
\begin{align*}
    \frac{\mathrm{d}}{\mathrm{d}u} m_0(u) \leqslant 0,\quad  \frac{\mathrm{d}}{\mathrm{d}u} \left(m_0(u) -\sqrt{\sum _{1\leqslant i \leqslant 3} m_i (u) ^2  } \right) \leqslant 0.
\end{align*}

In \cite{P}, Penrose introduced simple spacetimes in order to study gravitational waves in terms of conformal compactification and showed the peeling property of Weyl curvatures, see also \cite{NP1, NP2, NU}. 

Due to the loss property of energy-momentum for the zero cosmological constant, the Bondi energy-momentum can be referred to the total energy-momentum after loss carried away by gravitational waves. But it becomes sophisticated when the cosmological constant is nonzero, in particular, is positive which consists with the cosmological observations. It has been studied extensively for gravitational waves in this case in recent years, e.g. \cite{C, GLSWZ, HC, ABK1, ABK2, ABK3, ABK4, B, DH1, DH2, S1, S2}. Denote $\Lambda $ the cosmological constant. When $\Lambda \neq 0$, a detail asymptotic analysis of Bondi-Sachs metrics was provided \cite{GLSWZ}. Assuming \eqref{outgoing}, it yields
\begin{align*}
\beta = B + & O\left( \frac{1}{r^2}\right), \quad  U=X+O\left(\frac{1}{r}\right), \quad  W=Y+O\left(\frac{1}{r}\right), \\
 & \quad  V=O\left( r^3 \right)-2M(u, \theta, \phi)+O\left(\frac{1}{r}\right),
\end{align*}
where $B(u, \theta, \phi)$, $X(u, \theta, \phi)$, $Y(u, \theta, \phi)$ are boundary values of $\beta$, $U$, $W$ at $r \rightarrow \infty$. In particular, if $c$, $d$ are nonzero, $X$, $Y$ can not be zero. Fortunately, the peeling property still holds although these expansions are un-expected \cite{XZ}. In \cite{HC}, an alternative boundary condition was given in the axi-symmetric case, without assuming outgoing radiation condition, but deforming 2-sphere with
\[
\gamma=\Lambda f(u, \theta)+\frac{c(u, \theta)}{r}+O\Big(\frac{1}{r^3}\Big)
\]
and taking $B=X=Y=0$. In \cite{ABK1, ABK2, ABK3, ABK4}, asymptotics with $\Lambda >0$ was discussed in framework of conformal compactification, and the linearization theory as well as the quadrupole formula were derived. Some relevant works on the linearization theory can also be found in \cite{B, DH1, DH2}. The papers \cite{S1, S2} discussed the asymptotic vacuum and the electromagnetism Newman-Penrose equations as well as Bondi mass for $\Lambda \neq 0$. The boundary condition in \cite{S1, S2} is essentially equivalent to that given in \cite{HC}, where the peeling property was proved with this condition.

When the cosmological constant is zero, for the Einstein scalar field equations, ten conserved Newman–Penrose constants were defined in framework of Newman–Penrose formalism \cite{HC2}, and the Bondi-Sachs formalism was established \cite{LiZ}. Assuming \eqref{outgoing}, it yields
\begin{align*}
\beta = O\left(\frac{1}{r^2}\right), \,  U=O\left(\frac{1}{r^2}\right), \, W=O\left(\frac{1}{r^2}\right), \, V=r-2M(u, \theta, \phi)+O\left(\frac{1}{r}\right),
\end{align*}
and the scalar field
\begin{align*}
\Psi= O\left(\frac{1}{r}\right).
\end{align*}
These expansions result the peeling property \cite{LiZ}. Furthermore, denote 
\begin{align}
    l = c_\theta+2c\cot\theta+d_\phi\csc\theta, \quad \hat{l} =d_\theta+2d\cot\theta-c_\phi\csc\theta, \label{l-barl}\\
    \mathcal{M}(u, \theta, \phi)=M(u, \theta, \phi)-\frac{1}{2}\left(l_\theta+l\cot\theta+\hat{l}_\phi\csc\theta\right).\label{eq:NewMandOldMrelation}
\end{align}
In \cite{LiZ} (see, also \cite{HYZ}), $\mathcal{M}$ is used to define the Bondi energy-momentum
\begin{align*}
    m_\nu(u)=\frac{1}{4\pi}\int_{S^2}\mathcal{M}(u,\theta,\phi)n^\nu\mathrm{d}S, \quad \nu=0, 1, 2, 3.
\end{align*}
then the loss property of energy-momentum holds without assuming \eqref{hyz-u} \cite{LiZ, HYZ}
\begin{align*}
    \frac{\mathrm{d}}{\mathrm{d}u} m_0(u) \leqslant 0,\quad  
    \frac{\mathrm{d}}{\mathrm{d}u} \left(m _0(u) -\sqrt{\sum _{1\leqslant i \leqslant 3} m_i (u) ^2  } \right) \leqslant 0.
\end{align*}
If $c=d=0$ at some retarded time $u_0$, the positivity of the Bondi energy-momentum before $u_0$ holds for the Einstein scalar field equations \cite{LiZ}.

The Einstein scalar field equations for the Bondi-Sachs metrics with the zero cosmological constant can be referred to interaction of gravitational waves and dark matters, e.g. \cite{AGM, LiZ}. In this paper, we shall study them with the nonzero cosmological constant. This can be referred to interaction of gravitational waves, dark matters and dark energy when the cosmological constant is positive. Denote by $\nabla$ the Levi-Civita connection of $\mathbf{g}$. For massless scalar field $\Psi$, the energy-momentum tensor is given by
\begin{equation}\label{eq:scalarTDef}
    T_{\mu\nu}=\nabla_\mu\Psi\nabla_\nu\Psi-\frac{1}{2}\mathbf{g}_{\mu\nu}\nabla^\sigma\Psi\nabla_\sigma\Psi.
\end{equation}
The Einstein scalar field equations are
\begin{equation}
    R_{\mu\nu}=\Lambda \mathbf{g}_{\mu\nu}+\nabla_\mu\Psi\nabla_\nu\Psi. \label{esf}
\end{equation}
Assuming \eqref{outgoing}, we can finally prove $B=0$, therefore
\begin{align*}
\beta = O\left( \frac{1}{r^2}\right), \quad & U=\Lambda \tilde{X}\sin\theta+O\left(\frac{1}{r^2}\right), \quad  W=\Lambda \tilde{Y}\sin\theta+O\left(\frac{1}{r^2}\right), \\
 & V=O\left( r^3 \right)-2M(u, \theta, \phi)+O\left(\frac{1}{r}\right),
\end{align*}
and the scalar field
\begin{align*}
\Psi= \Lambda \tilde{I} +\frac{3 \tilde{I}_u}{r} +O\left(\frac{1}{r^3}\right),
\end{align*}
where $\tilde{I}=\tilde{I}(u)$ depends only on $u$, and $\tilde{X}$, $\tilde{Y}$ satisfy
\begin{align*}
    \tilde{X}_\theta \sin\theta-\tilde{Y}_\phi = \frac{2}{3}c, \quad
    \tilde{Y}_\theta \sin\theta+\tilde{X}_\phi = \frac{2}{3}d.
\end{align*}

With these expansions, we can prove the peeling property if $\tilde{I}=0$ in this case. We also derive a loss formula of the Bondi energy-momentum. Applying it to certain real data from gravitational waves, the loss formula is likely to indicate the loss property of the Bondi energy-momentum.

This paper is organized as follows.
In Section 2, we study the structure of the Einstein scalar field equations with the nonzero cosmological constant for Bondi-Sachs metrics and separate them into seven equations.
In Section 3, we provide asymptotic expansions of the Einstein scalar field equations.
In Section 4, we prove the peeling property for the Einstein scalar field equations if $\tilde{I}=0$.
In Section 5, we provide the Bondi energy-momentum for the nonzero cosmological constant, and derive a loss formula of the Bondi energy-momentum. We remark that, for certain real data from gravitational waves, the loss formula is likely to indicate the loss property of the Bondi energy-momentum.
In Appendix A, we provide explicit formulas relating to the Einstein scalar field equations.
In Appendix B, we provide explicit formulas for certain higher order coefficients in asymptotic expansions in Section 3.
In Appendix C, we provide explicit formulas for coefficients of $\Omega_{00}$ with respect to the cosmological constant in asymptotic expansions in Section 5.

\mysection{Einstein scalar field equations for $\Lambda \neq 0$}
\ls

In this section, we study the structure of the Einstein scalar field equations for Bondi-Sachs metric with the nonzero cosmological constant.

Let $\Psi$ be any smooth scalar field over a Bondi-Sachs spacetime. Denote
\begin{equation}
    \Omega_{\mu\nu} = R_{\mu\nu}-\Lambda\mathbf{g}_{\mu\nu}-\nabla_\mu\Psi\nabla_\nu\Psi, \quad \Omega=\mathbf{g}^{\mu\nu}\Omega_{\mu\nu}. \label{omega}
\end{equation}
By the twice contracted Bianchi identity, we have
\begin{equation}
    \nabla^\mu\bigg(\Omega_{\mu\nu}-\frac{\Omega}{2}\mathbf{g}_{\mu\nu}\bigg)
	=-\Box\Psi \nabla_\nu \Psi.  \label{omega1}
\end{equation}

We denote $x^0=u$, $x^1=r$, $x^2=\theta$, $x^3=\phi$ in \eqref{eq:BondiSachsMetric}. The Einstein scalar field equations with the nonzero cosmological constant are
\begin{equation}\label{eq:EinsteinTotalOmega}
    \Omega_{\mu\nu}=0, \quad \Box\Psi=0
\end{equation}
for $\mu,\,\nu=0,\,1,\,2,\,3$. Same as \cite{BBM, S, vdB, LiZ}, these ten equations are separated into three groups:
\begin{enumerate}
   \item six main equations
\begin{equation}\label{eq:main6eq}
    \Omega_{11}=\Omega_{12}=\Omega_{13}=\Omega_{22}=\Omega_{23}=\Omega_{33}=0,
\end{equation}
   \item one trivial equation
\begin{equation}\label{eq:trivial1eq}
    \Omega_{01}=0,
\end{equation}
   \item three supplementary equations
\begin{equation}\label{eq:supplementary3eq}
    \Omega_{00}=\Omega_{02}=\Omega_{03}=0.
\end{equation}
\end{enumerate}

For the Bondi-Sachs metric \eqref{eq:BondiSachsMetric},
\begin{equation}
    \mathbf{g}_{11}=\mathbf{g}_{12}=\mathbf{g}_{13}=0 \Longrightarrow \mathbf{g}^{00}=\mathbf{g}^{02}=\mathbf{g}^{03}=0, \label{g-zero}
\end{equation}
and the Christoffel symbols of the Bondi-Sachs metrics satisfy \cite{BBM, S, vdB, LiZ}
\begin{equation}\label{eq:Christoffel0g}
    \mathbf{g}^{\alpha\epsilon} \Gamma^0_{\alpha\epsilon}=\frac{2\mathrm{e}^{-2\beta}}{r}.
\end{equation}
Using \eqref{esf}, \eqref{omega1}, we obtain
\begin{equation}\label{eq:gChristoffelBianchi}
    \mathbf{g}^{\alpha\epsilon}
	\left( \frac{\partial \Omega_{\mu\alpha}}{\partial x^\epsilon}
	-\frac{1}{2}\frac{\partial \Omega_{\alpha\epsilon}}{\partial x^\mu}
	-\Gamma^\delta _{\alpha\epsilon} \Omega_{\mu\delta}\right)=0
\end{equation}
for $\mu=0,\,1,\,2,\,3$. 

In the following, we provide details that the six main equations \eqref{eq:main6eq} imply that the trivial equation \eqref{eq:trivial1eq} holds. Moreover, they imply that the three supplementary equations hold everywhere if they hold for at some $r_0>0$ \cite{BBM, S, vdB, LiZ}. Indeed, if we assume that \eqref{eq:main6eq} holds, then
\begin{enumerate}
   \item $\mu=1$: using \eqref{eq:main6eq} and \eqref{g-zero}, \eqref{eq:gChristoffelBianchi} gives
\begin{equation}\label{omega01}
    \mathbf{g}^{01} \frac{\partial \Omega _{10}}{\partial r}-\frac{\mathbf{g}^{01}}{2}\frac{\partial \Omega _{01}}{\partial r}-\frac{\mathbf{g}^{10}}{2}\frac{\partial \Omega _{10}}{\partial r}-\mathbf{g}^{\alpha\epsilon}\Gamma^0_{\alpha\epsilon}\Omega_{10}=-\mathbf{g}^{\alpha\epsilon}\Gamma^0_{\alpha\epsilon}\Omega_{01}=0.\\
\end{equation}
    \item $\mu=A$: using \eqref{eq:trivial1eq} and \eqref{eq:Christoffel0g}, \eqref{eq:gChristoffelBianchi} gives
\begin{equation}\label{eq:Omega0A}
    \frac{\mathrm{e}^{-2\beta}}{r^2}\frac{\partial}{\partial r}\left(r^2\Omega_{0A}\right) =0,\quad A=2,3.\\
\end{equation}
    \item $\mu=0$: using \eqref{eq:trivial1eq}, \eqref{eq:Christoffel0g} and
        $\Omega_{02}=\Omega_{03}=0$, \eqref{eq:gChristoffelBianchi} gives
\begin{equation}\label{eq:Omega00}
    \frac{\mathrm{e}^{-2\beta}}{r^2}\frac{\partial}{\partial r}\left(r^2\Omega_{00}\right) =0.
\end{equation}
\end{enumerate}
Therefore, \eqref{eq:trivial1eq} holds by \eqref{eq:Christoffel0g},
\eqref{omega01}, and  \eqref{eq:supplementary3eq} hold everywhere if they hold at some $r_0>0$ by \eqref{eq:Omega0A}, \eqref{eq:Omega00}. So it only needs to study the six main equations, which separate into two groups:
\begin{enumerate}
    \item four hypersurface equations
        \[ \Omega_{11}=\Omega_{12}=\Omega_{13}=0, \]
	\[ \left(\mathrm{e}^{-2\gamma}\Omega_{22}
	+\mathrm{e}^{2\gamma}\csc^2\theta\Omega_{33}\right)
	\cosh(2\delta)-2\csc\theta \Omega_{23}\sinh(2\delta)=0;  \]
    \item two standard equations
	\[ \mathrm{e}^{-2\gamma}\Omega_{22}-\mathrm{e}^{2\gamma}\csc^2\theta\Omega_{33}=0,  \]
	\[ \left(\mathrm{e}^{-2\gamma}\Omega_{22}
	+\mathrm{e}^{2\gamma}\csc^2\theta\Omega_{33}\right)\sinh(2\delta)
	-2\csc\theta \Omega_{23}\cosh(2\delta)=0.  \]
\end{enumerate}
We conclude that the Einstein scalar field equations \eqref{esf} are equivalent to the following seven equations:
\begin{align}
    \beta_r &=\mathscr{R}_1,  \label{eq:EQ1} \\
    \left(r^4\mathrm{e}^{-2\beta}\left(\mathrm{e}^{2\gamma}U_r\cosh(2\delta)+W_r\sinh(2\delta) \right)\right)_r
    &=\mathscr{R}_2, \label{eq:EQ2} \\
    \left(r^4\mathrm{e}^{-2\beta}
    \left(U_r\sinh(2\delta)+\mathrm{e}^{-2\gamma}W_r\cosh(2\delta)\right)\right)_r &=\mathscr{R}_3,  \label{eq:EQ3} \\
     V_r &=\mathscr{R}_4, \label{eq:EQ4} \\
       (r\gamma)_{ur}\cosh(2\delta)
    +2r(\gamma_u\delta_r+\delta_u\gamma_r)\sinh(2\delta)
    &=\mathscr{R}_5,  \label{eq:EQ5} \\
      (r\delta)_{ur}-2r\gamma_u\gamma_r\sinh(2\delta)\cosh(2\delta)
    &=\mathscr{R}_6,  \label{eq:EQ6} \\
     (r\Psi_u)_{r} &=\mathscr{R}_7, \label{eq:EQ7}
\end{align}
where $\mathscr{R}_1,\ldots,\mathscr{R}_7$ are given in Appendix A. The most important feature is that $\mathscr{R}_1,\ldots,\mathscr{R}_7$ do not contain any derivatives with respect to $x^0=u$.

\mysection{Asymptotic Expansions}
\ls

In this section, we study the power series solutions of the Einstein scalar field equations \eqref{eq:EQ1}-\eqref{eq:EQ7} for Bondi-Sachs metric \eqref{eq:BondiSachsMetric}. We follow the idea of \cite{BBM, S, vdB, LiZ} and assume that $\gamma$, $\delta$ satisfy the outgoing radiating condition
\begin{align}
\gamma &= \frac{c}{r}+\left(-\frac{1}{6}c^3-\frac{3}{2}d^2c+C\right)\frac{1}{r^3}+\frac{\overset{4}{\gamma}}{r^4}+O\left(\frac{1}{r^5}\right), \label{eq:gammaExpand} \\
\delta &= \frac{d}{r}+\left(-\frac{1}{6}d^3+\frac{1}{2}c^2d+D\right)\frac{1}{r^3}+\frac{\overset{4}{\delta}}{r^4}+O\left(\frac{1}{r^5}\right).\label{eq:deltaExpand}
\end{align}
The absence of $r^{-2}$ term in the above expansions prevent appearance of $\ln r $ term in the expansions of unknown in vacuum Einstein field equations \cite{BBM, S, vdB, LiZ}. We refer to \cite{CMS} for polyhomogeneity expansions involving $r^{-2}$ term for $\gamma$ and $\delta$.

For the scalar field $\Psi$, we take the following expansion \cite{LiZ} 
\begin{equation}\label{eq:PsiOriginal}
    \Psi = I(u, \theta, \phi)+\frac{H(u, \theta, \phi)}{r}+ \frac{K(u, \theta, \phi)}{r^2} 
    + \frac{L(u, \theta, \phi)}{r^3} 
    +  O\left(\frac{1}{r^4}\right).
\end{equation}

\begin{theorem}
Suppose that the Einstein scalar field equations for the nonzero cosmological constant \eqref{eq:EQ1}-\eqref{eq:EQ7} hold for Bondi-Sachs metric \eqref{eq:BondiSachsMetric}, and $c$, $d$, $\Psi$ satisfy \eqref{eq:gammaExpand}, \eqref{eq:deltaExpand}, \eqref{eq:PsiOriginal}. Then there is a coordinate transformations to preserve the metric, and, under the new coordinate systems (still denoted as $u, r, \theta, \phi$), the following asymptotic behaviors hold as $r \rightarrow \infty$
\begin{align}
    \beta &= -\frac{2(c^2+d^2)+9\tilde{I}_u^2}{8r^2} +\frac{\overset{4}{\beta}}{r^4}+O\left(\frac{1}{r^5}\right), \label{eq:betaExpand}\\
    U &= \Lambda \tilde{X}\sin\theta
    -\frac{l}{r^2} 
    +\frac{2\left(2cl+2d\hat{l}+3N(u,\theta,\phi)\right)}{3r^3} \notag \\
    &\quad +\frac{\overset{4}{U}}{r^4}
    +\frac{\overset{5}{U}}{r^5}
    +O\left(\frac{1}{r^6}\right), \label{eq:UExpand} \\
    W &= \Lambda \tilde{Y}\sin\theta
    -\frac{\hat{l}}{r^2}
    +\frac{2\left(2dl-2c\hat{l}+3P(u,\theta,\phi)\right)}{3r^3} \notag \\
    & \quad +\frac{\overset{4}{W}}{r^4}
    +\frac{\overset{5}{W}}{r^5}
    +O\left(\frac{1}{r^6}\right) , \label{eq:WExpand} \\
    V &= -\frac{\Lambda r^3}{3}
    +\Lambda \left(\tilde{X}_\theta\sin\theta +2\tilde{X}\cos\theta+\tilde{Y}_\phi\right)r^2 \notag \\
    &\quad +\frac{r}{4}\Bigg[\Lambda\left(2c^2+2d^2+ 9 \tilde{I}_u ^2\right)+4\Bigg] \label{eq:VExpand} \\
    &\quad -2\mathcal{M}(u, \theta, \phi)-\left(l_\theta+l\cot\theta+\hat{l}_\phi\csc\theta\right)
           +\frac{\overset{1}{V}}{r}+O\left(\frac{1}{r^2}\right), \notag \\
    \Psi & =\Lambda \tilde{I}+\frac{3 \tilde{I}_u }{r}+\frac{L}{r^3}+O\left(\frac{1}{r^4}\right). \label{eq:PsiExpand}  
\end{align}
where $l$, $\hat{l}$ are given by \eqref{l-barl}, $\tilde{I}=\tilde{I}(u)$ depends only on $u$, $\tilde{X}$, $\tilde{Y}$ satisfy
\begin{equation}
    \tilde{X}_\theta \sin\theta-\tilde{Y}_\phi = \frac{2}{3}c, \quad
    \tilde{Y}_\theta \sin\theta+\tilde{X}_\phi = \frac{2}{3}d, \label{eq:c-dRestrictTilde}
\end{equation}
and $\overset{4}{\beta}$, $\overset{4}{U}$, $\overset{5}{U}$, $\overset{4}{W}$, $\overset{5}{W}$, $\overset{1}{V}$ are given in Appendix B.
\end{theorem}
\begin{proof}
    From \eqref{eq:EQ1}-\eqref{eq:EQ4}, we can derive the following asymptotic expansions under \eqref{eq:gammaExpand}, \eqref{eq:deltaExpand}, \eqref{eq:PsiOriginal} (see also \cite{GLSWZ})
\begin{align*}
    \beta &= B-\frac{2(c^2+d^2)+H^2}{8r^2}+O\left(\frac{1}{r^3}\right),\\ 
    U &= X+\frac{2\mathrm{e}^{2B}B_\theta}{r}+O\left(\frac{1}{r^2}\right),\\
    W &= Y+\frac{2\mathrm{e}^{2B}B_\phi\csc\theta}{r}+O\left(\frac{1}{r^2}\right),\\
    V &=-\frac{1}{3}\mathrm{e}^{2B}\Lambda r^3 + \bigg(X_\theta +X\cot\theta+Y_\phi\csc\theta \bigg)r^2+O(r),
\end{align*}
where $B(u, \theta, \phi)$, $X(u, \theta, \phi)$, $Y(u, \theta, \phi)$ are boundary values of $\beta$, $U$, $W$ as $r \rightarrow \infty$. 

In the following we show that we can take more general coordinate transformation than that in \cite{GLSWZ} to reduce $B=0$. Indeed, let
\begin{equation}\label{eq:coordinateTransform}
\left\{    
    \begin{aligned}
        u &= u_0(\hat{u},\hat{\theta},\hat{\phi})+\frac{u_1(\hat{u},\hat{\theta},\hat{\phi})}{\hat{r}}
	+O\left(\frac{1}{\hat{r}^2}\right), \\
        r &= r_{-1}(\hat{u},\hat{\theta},\hat{\phi})\hat{r}+r_0(\hat{u},\hat{\theta},\hat{\phi})+O\left(\frac{1}{\hat{r}}\right), \\
        \theta &= \hat{\theta}+\frac{\theta_1(\hat{u},\hat{\theta},\hat{\phi})}{\hat{r}}+O\left(\frac{1}{\hat{r}^2}\right), \\
        \phi &= \hat{\phi}+\frac{\phi_1(\hat{u},\hat{\theta},\hat{\phi})}{\hat{r}}+O\left(\frac{1}{\hat{r}^2}\right). 
    \end{aligned} 
\right.
\end{equation}
Denote $\hat{\mathbf{g}}_{\mu\nu} =\mathbf{g}(\partial_\mu, \partial_\nu)$. We have
\begin{align*}
    \hat{\mathbf{g}}_{01} &=
    -r_{-1}u_{0,\hat{u}}\Bigg\{ 
    \frac{1}{3}\Lambda r_{-1} u_1\mathrm{e}^{4B}
    +\mathrm{e}^{2B}
    +r_{-1}X\left( u_1 X -\theta_1\right)  \\
    &\quad 
    +r_{-1}Y\left(u_1 Y -\phi_1\sin\hat{\theta}\right)
    \Bigg\}
     +O\left(\frac{1}{\hat{r}}\right)  \\
    \hat{\mathbf{g}}_{11} &= \frac{r_{-1}}{\hat{r}^2} \Bigg\{ \frac{1}{3}\Lambda r_{-1} u_1^2\mathrm{e}^{4B}
    +2u_1\mathrm{e}^{2B}
    +r_{-1}\left( u_1X-\theta_1\right)^2
     \\
    &\quad 
    +r_{-1}\left(u_1Y-\phi_1\sin\hat{\theta}\right)^2
    \Bigg\} +O\left(\frac{1}{\hat{r}^3}\right)  \\
    \hat{\mathbf{g}}_{12} &= r_{-1}\Bigg\{ -\frac{1}{3}\Lambda r_{-1} u_{0,\hat{\theta}} u_1\mathrm{e}^{4B}
    -u_{0,\hat{\theta}}\mathrm{e}^{2B}
    - r_{-1} Y u_{0,\hat{\theta}} \left(u_1Y-\phi_1\sin\hat{\theta}\right) \\
    &\quad 
    -r_{-1}\left(u_1X-\theta_1\right)\left(u_{0,\hat{\theta}}X-1\right) \Bigg\} +O\left(\frac{1}{\hat{r}}\right) \\
    \hat{\mathbf{g}}_{13} &=
    r_{-1}\Bigg\{ -\frac{1}{3}\Lambda r_{-1} u_{0,\hat{\phi}} u_1\mathrm{e}^{4B}
    -u_{0,\hat{\phi}}\mathrm{e}^{2B}
    -r_{-1}X u_{0,\hat{\phi}}\left( u_1 X -\theta_1\right) \\
    &\quad -r_{-1}\left( u_1Y-\phi_1\sin\hat{\theta}\right)
    \left(u_{0,\hat{\phi}}Y-\sin\hat{\theta}\right)
    \Bigg\}  +O\left(\frac{1}{\hat{r}}\right)  
\end{align*}
\begin{equation}\label{eq:cross}
\begin{aligned}
    \hat{\mathbf{g}}_{22}&\hat{\mathbf{g}}_{33}-\hat{\mathbf{g}}_{23}^2-\hat{r}^4\sin^2\hat{\theta} \\
    &=\hat{r}^4\Bigg\{\frac{1}{3}r_{-1}^4\bigg[ 
        \Lambda \mathrm{e}^{4B}\left(u_{0,\hat{\theta}}^2\sin^2\hat{\theta}
        +u_{0,\hat{\phi}}^2 \right) \\
    &\quad 
        +3\left(X u_{0,\hat{\theta}} \sin\hat{\theta}
        + Yu_{0,\hat{\phi}}-\sin\hat{\theta}\right)^2
    \bigg]  -\sin^2\hat{\theta}  
    \Bigg\}  \\
    &\quad +\hat{r}^3 \Bigg\{ \frac{4}{3}r_0 r_{-1}^3\bigg[ \mathrm{e}^{4B}\Lambda \left(u_{0,\hat{\theta}}^2 \sin^2\hat{\theta}
    + u_{0,\hat{\phi}}^2 \right) \\
    &\quad 
    +3 \left(X u_{0,\hat{\theta}} \sin\hat{\theta}+Yu_{0,\hat{\phi}}-\sin\hat{\theta}\right)^2  \bigg]  + \Theta
    \Bigg\}
    +O\left(\bar{r}^2\right),
\end{aligned}
\end{equation}
where $\Theta$ is given in the Appendix B, which doesn't contain $r_0$.

Comparing the $\hat{r}^4$-coefficient in \eqref{eq:cross}, we can obtain
\[ r_{-1}=\left(\frac{3}{\Lambda \mathrm{e}^{4B}\left(u_{0,\hat{\theta}}^2
        +u_{0,\hat{\phi}}^2\csc^2\hat{\theta} \right) 
        +3\left(X u_{0,\hat{\theta}} 
        + Yu_{0,\hat{\phi}}\csc\hat{\theta}-1\right)^2}\right)^{\frac{1}{4}} >0. \]
Using $\hat{\mathbf{g}}_{11}=\hat{\mathbf{g}}_{12}=\hat{\mathbf{g}}_{13}=0$, we obtain
\begin{align*}
    0 &= \frac{1}{3}\Lambda r_{-1} u_1^2 \mathrm{e}^{4B}
    +2 u_1\mathrm{e}^{2B}
    +r_{-1} \left( u_1X-\theta_1\right)^2
    +r_{-1} \left( u_1Y -\phi_1\sin\hat{\theta}\right)^2, \\
    0 &=-\frac{1}{3}\Lambda r_{-1} u_{0,\hat{\theta}} u_1\mathrm{e}^{4B}
    -u_{0,\hat{\theta}}\mathrm{e}^{2B}
    - r_{-1} Y u_{0,\hat{\theta}} \left(u_1Y-\phi_1\sin\hat{\theta}\right),  \\
    &\quad 
    -r_{-1}\left(u_1X-\theta_1\right)\left(u_{0,\hat{\theta}}X-1\right) \\
    0 &= -\frac{1}{3}\Lambda r_{-1} u_{0,\hat{\phi}} u_1\mathrm{e}^{4B}
    -u_{0,\hat{\phi}}\mathrm{e}^{2B}
    -r_{-1}X u_{0,\hat{\phi}}\left( u_1 X -\theta_1\right) \\
    &\quad -r_{-1}\left( u_1Y-\phi_1\sin\hat{\theta}\right)
    \left(u_{0,\hat{\phi}}Y-\sin\hat{\theta}\right),
\end{align*}
thus we can solve $u_1$, $\theta_1$, $\phi_1$ for given $u_0$, $r_{-1}$. Now we choose
\[ u_{0,\hat{u}}= \Bigg\{ \frac{1}{3}\Lambda r_{-1} ^2 u_1\mathrm{e}^{4B}
    +r_{-1}\mathrm{e}^{2B}
    +r_{-1} ^2 X\left( u_1 X -\theta_1\right)  
    +r_{-1} ^2 Y\left(u_1 Y -\phi_1\sin\hat{\theta}\right)\Bigg\}^{-1},  \]
we obtain
\[ \hat{B} = 0 \]
Finally, comparing the $\hat{r}^3$-coefficient in \eqref{eq:cross}, we obtain
\begin{align*}
r_0= -\frac{3\Theta}{4 r_{-1}^3}\bigg[ \mathrm{e}^{4B}\Lambda \left(u_{0,\hat{\theta}}^2 \sin^2\hat{\theta}
    + u_{0,\hat{\phi}}^2 \right) 
    +3 \left(X u_{0,\hat{\theta}} \sin\hat{\theta}+Yu_{0,\hat{\phi}}-\sin\hat{\theta}\right)^2  \bigg]^{-1}.
\end{align*}

It is straightforward that the right hand sides of \eqref{eq:EQ5}, \eqref{eq:EQ6} are
\begin{align*}
    \mathscr{R}_5 &= \frac{2\Lambda c+3X\cot\theta-3X_\theta+3Y_\phi\csc\theta}{6}
+\frac{I_\theta^2-I_\phi^2\csc^2\theta}{4r}+O\left(\frac{1}{r^2}\right), \\
    \mathscr{R}_6 &= \frac{2\Lambda d+3Y\cot\theta-3Y_\theta-3X_\phi\csc\theta}{6}+\frac{I_\theta I_\phi \csc\theta}{2r}
+O\left(\frac{1}{r^2}\right).
\end{align*}
As the left hand sides of \eqref{eq:EQ5}, \eqref{eq:EQ6} are both $O\left(\frac{1}{r^3}\right)$, we can obtain
\begin{align*}
    2\Lambda c+3X\cot\theta-3X_\theta+3Y_\phi\csc\theta &=0,\\
    2\Lambda d+3Y\cot\theta-3Y_\theta-3X_\phi\csc\theta &=0,\\
    I_\theta^2-I_\phi^2\csc^2\theta &=0,\\
    I_\theta I_\phi \csc\theta &=0.
\end{align*}
Denote $X=\Lambda\tilde{X}\sin\theta$, $W=\Lambda\tilde{Y}\sin\theta$. Then \eqref{eq:c-dRestrictTilde} follows. Moreover, 
\begin{align*}
I _\theta =I _\phi =0, \quad I=I(u).
\end{align*}
With this condition, the right hand side of \eqref{eq:EQ7} is
\begin{equation*}
    \mathscr{R}_7 =\frac{\Lambda H}{3}+\frac{\Lambda K}{3r}+O\left(\frac{1}{r^2}\right),
\end{equation*}
while the left hand side of \eqref{eq:EQ7} is
    \[ (r\Psi_u)_r=I_u-\frac{K_u}{r^2}-\frac{2L_u}{r^3}+O\left(\frac{1}{r^4}\right). \]
Denote $I=\Lambda \tilde{I}$. This gives that
\begin{align*}
K=0, \quad H = 3 \tilde{I}_u.
\end{align*}
We can obtain the asymptotic expansions \eqref{eq:betaExpand}-\eqref{eq:PsiExpand}.
\end{proof}

\mysection{Peeling Property}
\ls

In this section, we prove the peeling property for the Einstein scalar field equations with the nonero cosmological constant.

Denote $\mathrm{i}=\sqrt{-1}$. The Bondi-Sachs metric \eqref{eq:BondiSachsMetric} has the following null tetrad \cite{XZ}
\begin{equation}\label{eq:NPbasis}
\begin{aligned}
    \hat{e}_0 &=\boldsymbol{l}=\mathrm{e}^{-2\beta}\left(\frac{\partial}{\partial u}-\frac{V}{2r}\frac{\partial}{\partial r}+U\frac{\partial}{\partial\theta}+W\csc\theta \frac{\partial}{\partial \phi}\right), \\
    \hat{e}_1 &=\boldsymbol{k} =\frac{\partial}{\partial r}, \\
    \hat{e}_2 &=\boldsymbol{m} =\frac{\mathrm{e}^{-\gamma}\big(1-\mathrm{i}\sinh(2\delta)\big)}{r\sqrt{2\cosh(2\delta)}}\frac{\partial}{\partial \theta}
    +\frac{\mathrm{i}\mathrm{e}^\gamma \sqrt{\cosh(2\delta)}\csc\theta}{\sqrt{2}r}\frac{\partial}{\partial \phi}, \\
    \hat{e}_3 &=\bar{\boldsymbol{m}}=\frac{\mathrm{e}^{-\gamma}\big(1+\mathrm{i}\sinh(2\delta)\big)}{r\sqrt{2\cosh(2\delta)}}\frac{\partial}{\partial\theta}
    -\frac{\mathrm{i}\mathrm{e}^\gamma\sqrt{\cosh(2\delta)}\csc\theta}{\sqrt{2}r}\frac{\partial}{\partial \phi}.
\end{aligned}
\end{equation}
Under this null tetrad, both the metric matrix $(\hat{\mathbf{g}}_{\mu\nu})=(\mathbf{g}(\hat{e}_\mu,\hat{e}_\nu))$ and its inverse matrix $(\hat{\mathbf{g}}^{\mu\nu})$ are
\begin{equation*}
    \begin{pmatrix}
        0 & -1 & 0 & 0 \\
        -1 & 0 & 0 & 0 \\
        0 & 0 & 0 & 1 \\
        0 & 0 & 1 & 0
    \end{pmatrix}.
\end{equation*}
The structure coefficients $\hat{C}_{\mu\nu}^\sigma$ of the tetrad satisfy
\begin{equation*}
   [\hat{e}_\mu,\hat{e}_\nu]=\hat{C}_{\mu\nu}^\sigma \hat{e}_\sigma, \quad \hat{C}_{\mu\nu\sigma}=\hat{C}_{\mu\nu}^\tau \hat{\mathbf{g}}_{\tau\sigma}.
\end{equation*}
Denote $\langle X,Y\rangle=\mathbf{g}(X,Y)$. Use the Koszul formula \cite{O0}
\begin{align*}
	2\langle \nabla_X Y,Z\rangle
    &=X\langle Y,Z\rangle+Y\langle Z,X\rangle-Z\langle X,Y\rangle \\
	&\quad +\langle [X,Y],Z\rangle -\langle [Y,Z],X\rangle
	+\langle [Z,X],Y\rangle,
\end{align*}
we obtain connection coefficients
\begin{equation*}
\hat{\Gamma}_{\mu\nu\sigma}=\big \langle \nabla_{\hat{e}_\mu}\hat{e}_\nu,\hat{e}_\sigma  \big\rangle
=\frac{1}{2}\left(\hat{C}_{\mu\nu\sigma}-\hat{C}_{\nu\sigma\mu}+\hat{C}_{\sigma\mu\nu}\right).
\end{equation*}

The spin coefficients of the null tetrad \eqref{eq:NPbasis} are twelve complex-valued functions, which are given as follows \cite{O}.
\begin{align*}
    \boldsymbol{\kappa} &=-\langle \nabla_{\boldsymbol{k}} \boldsymbol{k},\boldsymbol{m}\rangle=-\hat{\Gamma}_{112}, \\
    \boldsymbol{\rho} &=-\langle \nabla_{\bar{\boldsymbol{m}}}\boldsymbol{k},\boldsymbol{m}\rangle=-\hat{\Gamma}_{312}, \\
    \boldsymbol{\sigma} &=-\langle \nabla_{\boldsymbol{m}} \boldsymbol{k},\boldsymbol{m}\rangle=-\hat{\Gamma}_{212}, \\
    \boldsymbol{\tau} &=-\langle \nabla_{\boldsymbol{l}} \boldsymbol{k},\boldsymbol{m}\rangle=-\hat{\Gamma}_{012},\\
    \boldsymbol{\nu} &=\langle \nabla_{\boldsymbol{l}} \boldsymbol{l},\bar{\boldsymbol{m}}\rangle=\hat{\Gamma}_{003}, \\
    \boldsymbol{\mu} &=\langle\nabla_{\boldsymbol{m}} \boldsymbol{l},\bar{\boldsymbol{m}}\rangle=\hat{\Gamma}_{203}, \\
    \boldsymbol{\lambda} &=\langle\nabla_{\bar{\boldsymbol{m}}}\boldsymbol{l},\bar{\boldsymbol{m}}\rangle=\hat{\Gamma}_{303}, \\
    \boldsymbol{\pi} &=\langle\nabla_{\boldsymbol{k}} \boldsymbol{l},\bar{\boldsymbol{m}}\rangle=\hat{\Gamma}_{103},\\
    \boldsymbol{\varepsilon} &=\frac{1}{2}
    \Big(-\langle \nabla_{\boldsymbol{k}}\boldsymbol{k},\boldsymbol{l}\rangle+\langle \nabla_{\boldsymbol{k}} \boldsymbol{m},\bar{\boldsymbol{m}}\rangle\Big)
=\frac{1}{2}\Big(-\hat{\Gamma}_{110}+\hat{\Gamma}_{123}\Big), \\
   \boldsymbol{\gamma} &=\frac{1}{2}\Big(\langle\nabla_{\boldsymbol{l}} \boldsymbol{l},\boldsymbol{k}\rangle-\langle\nabla_{\boldsymbol{l}}\bar{\boldsymbol{m}},\boldsymbol{m}\rangle \Big)=\frac{1}{2}\left(\hat{\Gamma}_{001}-\hat{\Gamma}_{032}\right), \\
   \boldsymbol{\beta} &=\frac{1}{2}
    \Big(-\langle \nabla_{\boldsymbol{m}}\boldsymbol{k},\boldsymbol{l}\rangle+\langle\nabla_{\boldsymbol{m}} \boldsymbol{m},\bar{\boldsymbol{m}}\rangle\Big)
=\frac{1}{2}
\left(-\hat{\Gamma}_{210}+\hat{\Gamma}_{223}\right), \\
    \boldsymbol{\alpha} &=\frac{1}{2}\Big(\langle \nabla_{\bar{\boldsymbol{m}}}\boldsymbol{l},\boldsymbol{k}\rangle-\langle \nabla_{\bar{\boldsymbol{m}}}\bar{\boldsymbol{m}},\boldsymbol{m}\rangle\Big)=\frac{1}{2}\left(\hat{\Gamma}_{301}-\hat{\Gamma}_{332}\right).
\end{align*}

For Bondi-Sachs metric \eqref{eq:BondiSachsMetric}, we have
\begin{align*}
    \boldsymbol{\kappa} &=0,\\
    \boldsymbol{\rho}   &=-\frac{1}{r},\\
    \boldsymbol{\sigma} &=-\gamma_r-\tanh(2\delta)\delta_r+\mathrm{i}\bigg(\sinh(2\delta)\gamma_r-\operatorname{sech}(2\delta)\delta_r\bigg),\\
    \boldsymbol{\tau}   &=\frac{1}{2\sqrt{2}r\sqrt{\cosh(2\delta)}}
    \bigg( -2 \mathrm{e}^{-\gamma}\beta_\theta +\mathrm{e}^{-2\beta +\gamma}r^2 U_r\cosh(2\delta)\\
    &\quad +\mathrm{e}^{-2\beta -\gamma}r^2 W_r\sinh(2\delta)+\mathrm{i}\Big( -2\mathrm{e}^\gamma\beta_\phi\cosh(2\delta)\csc\theta\\
    &\quad +2\mathrm{e}^{-\gamma}\beta_\theta\sinh(2\delta)+\mathrm{e}^{-2\beta-\gamma}r^2 W_r\Big) \bigg),  \\
    \boldsymbol{\nu} &=\frac{1}{2\sqrt{2} r^2\sqrt{\cosh(2\delta)} }\bigg(\mathrm{e} ^{-2\beta-\gamma} V_\theta
      +\mathrm{i}\Big(\mathrm{e}^{-2\beta-\gamma}V_\theta\sinh(2\delta)\\
    &\quad -\mathrm{e}^{-2\beta+\gamma}V_\phi \csc\theta \cosh(2\delta) \Big)\bigg),  \\ 
    \boldsymbol{\mu} &=\frac{\mathrm{e}^{-2\beta}}{2}\bigg(U\cot\theta+ U_\theta+W_\phi\csc\theta -\frac{V}{r^2} \bigg),  \\
    \boldsymbol{\lambda} &=\frac{\mathrm{e}^{-2\beta}}{8}
    \bigg(8  \mathrm{e}^{-2\gamma} \tanh(2\delta) (W_\theta - W\cot\theta)\\
    &\quad -\frac{4V}{r}(\delta_r\tanh(2\delta)+\gamma_r) \\
    &\quad  -4 U(\cot\theta-2\delta_\theta \tanh(2\delta)-2\gamma_\theta)\\
    &\quad  +8 W\csc\theta ( \delta_\phi \tanh(2\delta)+\gamma_\phi) \\
	&\quad  -4 (W_\phi\csc\theta-U_\theta-2\delta_u \tanh(2\delta)-2\gamma_u ) \\
    &\quad  -4\mathrm{i} \Big( \cosh(2\delta)  \mathrm{e}^{2\gamma}U_\phi \csc\theta \\
    &\quad  -\left(2\operatorname{sech}(2\delta)-\cosh(2\delta)\right)\mathrm{e}^{-2\gamma}( W\cot\theta-W_\theta) \\
    &\quad  +\operatorname{sech}(2\delta) \Big(2U \delta _\theta +2W \delta _\phi \csc\theta  -\frac{V}{r} \delta_r +2\delta _u \Big)\\
    &\quad  +\sinh(2\delta) \cdot\Big(U \cot\theta-2 U \gamma_\theta -U_\theta \\
    &\quad -2W \gamma _\phi \csc\theta +W _\phi \csc\theta +\frac{V}{r}\gamma _r -2\gamma_u\Big)\Big)\bigg),\\      
    \boldsymbol{\pi} &= \frac{\mathrm{e}^{-2\beta-\gamma}}{2\sqrt{2}r \sqrt{\cosh(2\delta)}}
    \bigg( 2\mathrm{e}^{2\beta}\beta_\theta+\mathrm{e}^{2\gamma}r^2 U_r \cosh(2\delta) \\
    &\quad +r^2 W_r\sinh(2\delta) -\mathrm{i}\Big( 2\mathrm{e}^{2\gamma+2\beta}\beta_\phi\csc\theta \cosh(2\delta)\\
    &\quad  -2\mathrm{e}^{2\beta}\beta_\theta \sinh(2\delta)+r^2 W_r\Big)\bigg),\\
    \boldsymbol{\varepsilon} &= \beta_r+\frac{\mathrm{i}\operatorname{sech}(2\delta)}{4}\bigg( \sinh(4\delta)\gamma_r-2\delta_r \bigg),  \\
    \boldsymbol{\gamma} &= \frac{\mathrm{e}^{-2\beta}}{8}\bigg( -\frac{2V}{r^2}+\frac{2V_r}{r}  \\
    &\quad -2\mathrm{i} \Big( \cosh(2\delta)  \Big(  \mathrm{e}^{2\gamma} U_\phi \csc\theta  + \mathrm{e}^{-2\gamma}(W\cot\theta-W_\theta)\Big)\\
    &\quad +\operatorname{sech}(2\delta) \Big(2U \delta _\theta +2W \delta _\phi \csc\theta   -\frac{V}{r} \delta_r +2\delta _u \Big)\\
    &\quad +\sinh(2\delta) \cdot\Big(U \cot\theta-2 U \gamma_\theta -U_\theta \\
    &\quad -2W \gamma _\phi \csc\theta +W _\phi \csc\theta +\frac{V}{r}\gamma _r -2\gamma_u\Big)\Big)\bigg),\\
    \boldsymbol{\beta} &= \frac{\mathrm{e}^{-\gamma}}{4\sqrt{2}r \sqrt{\cosh(2\delta)}}\bigg( 2\beta_\theta
	-2\gamma_\theta +2\cot\theta -2\delta_\theta \tanh(2\delta)\\
    &\quad +r^2 \mathrm{e}^{2\gamma-2\beta}U_r \cosh(2\delta)+\mathrm{e}^{-2\beta}r^2 W_r\sinh(2\delta)\\
    &\quad +2\mathrm{i}\bigg( \cosh(2\delta) \Big(\mathrm{e}^{2\gamma}\beta_\phi\csc\theta +\mathrm{e}^{2\gamma}\gamma _\phi\csc\theta -2\delta _\theta \Big)\\
    &\quad +\frac{\mathrm{e}^{-2\beta}}{2}r^2 W_r +\delta_\theta \tanh(2 \delta) \sinh(2\delta)\\
    &\quad +\sinh(2\delta)\Big( \mathrm{e}^{2\gamma}\delta_\phi\csc\theta -\beta_\theta +\gamma_\theta -\cot\theta \Big)\bigg) \bigg),\\
    \boldsymbol{\alpha} &= \frac{\mathrm{e}^{-\gamma}}{4\sqrt{2}r \sqrt{\cosh(2\delta)}}\bigg( 2\beta_\theta
	+2\gamma_\theta -2\cot\theta +2\delta_\theta \tanh(2\delta)\\
    &\quad +r^2 \mathrm{e}^{2\gamma-2\beta}U_r \cosh(2\delta)+\mathrm{e}^{-2\beta}r^2 W_r\sinh(2\delta)\\
    &\quad +2\mathrm{i}\bigg( \cosh(2\delta) \Big(-\mathrm{e}^{2\gamma}\beta_\phi\csc\theta +\mathrm{e}^{2\gamma}\gamma _\phi\csc\theta -2\delta _\theta \Big)\\
    &\quad -\frac{\mathrm{e}^{-2\beta}}{2}r^2 W_r +\delta_\theta \tanh(2 \delta) \sinh(2\delta)\\
    &\quad +\sinh(2\delta)\Big( \mathrm{e}^{2\gamma}\delta_\phi\csc\theta +\beta_\theta +\gamma_\theta -\cot\theta \Big)\bigg) \bigg).
\end{align*}
It holds that
\begin{equation*}
    \bar{\boldsymbol{\mu}}=\boldsymbol{\mu}, \quad \boldsymbol{\alpha}+\bar{\boldsymbol{\beta}}=\boldsymbol{\pi}.
\end{equation*}

Denote by $\mathcal{W}$ the $(0,4)$ type Weyl tensor. The Weyl scalars are \cite{O}
\begin{align*}
    \Phi _0 =\mathcal{W}_{\boldsymbol{k}\boldsymbol{m}\boldsymbol{k}\boldsymbol{m}}, \,\,
    \Phi_1 =\mathcal{W}_{\boldsymbol{k}\boldsymbol{l}\boldsymbol{k}\boldsymbol{m}}, \,\,
    \Phi_2 =\mathcal{W}_{\boldsymbol{k}\boldsymbol{m}\bar{\boldsymbol{m}}\boldsymbol{l}}, \,\,
    \Phi_3 =\mathcal{W}_{\boldsymbol{l}\boldsymbol{k}\boldsymbol{l}\bar{\boldsymbol{m}}}, \,\,
    \Phi_4 =\mathcal{W}_{\boldsymbol{l}\bar{\boldsymbol{m}}\boldsymbol{l}\bar{\boldsymbol{m}}}.
\end{align*}

Now we prove the peeling property.
\begin{theorem}
Suppose that the Einstein scalar field equations for the nonero cosmological constant \eqref{eq:EQ1}-\eqref{eq:EQ7} hold for Bondi-Sachs metric \eqref{eq:BondiSachsMetric}, and $c$, $d$, $\Psi$ satisfy \eqref{eq:gammaExpand}, \eqref{eq:deltaExpand}, \eqref{eq:PsiOriginal}. If $\tilde{I}_u=0$, then the peeling property holds
\begin{align*}
\Phi_k=\frac{f_k(u,\theta, \phi)}{r^{5-k}}+O\left(\frac{1}{r^{6-k}}\right),\quad k=0,\dots,4.
\end{align*}
\end{theorem}
\begin{proof}
    It is straightforward that
\begin{align*}
    \Phi_0 &= \boldsymbol{k}(\boldsymbol{\sigma})
    -\boldsymbol{\sigma}(2\boldsymbol{\rho}+3\boldsymbol{\varepsilon}-\bar{\boldsymbol{\varepsilon}}),  \\
    \Phi_1 &= \boldsymbol{k}(\boldsymbol{\beta})-\boldsymbol{m}(\boldsymbol{\varepsilon})
    -(\boldsymbol{\alpha}+\boldsymbol{\pi})\boldsymbol{\sigma}
    -(\boldsymbol{\rho}+\boldsymbol{\varepsilon}-\bar{\boldsymbol{\varepsilon}})\boldsymbol{\beta}, \\
    \Phi_2 &= \boldsymbol{k}(\boldsymbol{\mu})-\boldsymbol{m}(\boldsymbol{\pi})-\frac{\Lambda}{3}
    +\frac{\boldsymbol{l}(\Psi)\boldsymbol{k}(\Psi)}{6}
    -\frac{\boldsymbol{m}(\Psi)\bar{\boldsymbol{m}}(\Psi)}{6}-\boldsymbol{\rho}\boldsymbol{\mu}\\
    &\quad -\boldsymbol{\sigma}\boldsymbol{\lambda}-2\boldsymbol{\beta}\boldsymbol{\pi}
    +(\boldsymbol{\varepsilon}+\bar{\boldsymbol{\varepsilon}})\boldsymbol{\mu},  \\
    \Phi_3 &=\boldsymbol{k}(\boldsymbol{\nu})-\boldsymbol{l}(\boldsymbol{\pi})
    -\boldsymbol{\mu}(\boldsymbol{\pi}+\bar{\boldsymbol{\tau}})
    -\boldsymbol{\lambda}(\bar{\boldsymbol{\pi}}+\boldsymbol{\tau})
    -\boldsymbol{\pi}(\boldsymbol{\gamma}-\bar{\boldsymbol{\gamma}}) \\
    &\quad +\boldsymbol{\nu}(3\boldsymbol{\varepsilon}+\bar{\boldsymbol{\varepsilon}})
     -\frac{1}{2}\boldsymbol{l}(\Psi)\bar{\boldsymbol{m}}(\Psi),\\
    \Phi_4 &= -\boldsymbol{l}(\boldsymbol{\lambda})+\bar{\boldsymbol{m}}(\boldsymbol{\nu})
         -(2\boldsymbol{\mu}+3\boldsymbol{\gamma}-\bar{\boldsymbol{\gamma}})\boldsymbol{\lambda}
         +(2\boldsymbol{\alpha}+2\boldsymbol{\pi}-\bar{\boldsymbol{\tau}})\boldsymbol{\nu}.
\end{align*}
Same as that in \cite{XZ}, using the asymptotic expansions \eqref{eq:betaExpand}-\eqref{eq:PsiExpand}, in particular, using \eqref{eq:c-dRestrictTilde}, we obtain
\begin{align*}
\Phi_k=O\left(\frac{1}{r^{5-k}}\right)\,\,(k \neq 2), \quad \Phi_2 =\frac{\Lambda \tilde{I}_u^2}{2r^2} +O\left(\frac{1}{r^{3}}\right).
\end{align*}
Therefore, if $\tilde{I}_u=0$,  the theorem follows. 
\end{proof}

\mysection{The Bondi Energy-Momentum}
\ls

\allowdisplaybreaks

In this section we study the Bondi energy-momentum for the Einstein scalar field equations with the nonzero cosmological constant. 

Under the asymptotic expansions \eqref{eq:betaExpand}-\eqref{eq:PsiExpand}, we have
\begin{align*}
 \Omega_{00} &= \overset{-1}{\Omega} _{00} r+\overset{0}{\Omega}_{00}+\frac{\overset{1}{\Omega}_{00}}{r}  +\frac{\overset{2}{\Omega}_{00}}{r^2}+O\left(\frac{1}{r^3}\right),\\
    \overset{2}{\Omega}_{00}
    &=-2c_u^2-2d_u^2-9\tilde{I}_{uu}^2-2\mathcal{M}_u \\
    &\quad +\Lambda\bigg(-6\mathcal{M}\tilde{X}\cos\theta-4\mathcal{M}_\phi\tilde{Y}
-3\mathcal{M}\tilde{Y}_\phi \\
    &\quad -4\mathcal{M}_\theta\tilde{X}\sin\theta
    -3\mathcal{M}\tilde{X}_\theta \sin\theta+ \tensor[_{\Lambda}]{\overset{2}{\Omega}}{_{00}}\bigg) \\
    &\quad
    +\Lambda^2\bigg(2\mathcal{M} c\tilde{X}^2\sin^2\theta+ 4\mathcal{M}d\tilde{X}\tilde{Y}\sin^2\theta \\
    &\quad -2\mathcal{M}c\tilde{Y}^2\sin^2\theta + \tensor[_{\Lambda^2}]{\overset{2}{\Omega}}{_{00}}\bigg)
+\Lambda^3 \tensor[_{\Lambda^3}]{\overset{2}{\Omega}}{_{00}} ,
\end{align*}
where $\tensor[_{\Lambda}]{\overset{2}{\Omega}}{_{00}}$, $\tensor[_{\Lambda^2}]{\overset{2}{\Omega}}{_{00}}$, $\tensor[_{\Lambda^3}]{\overset{2}{\Omega}}{_{00}}$ are given in the Appendix C, which depend on $c,d,C,D,\tilde{X},\tilde{Y},\tilde{I},\overset{4}{\gamma},\overset{4}{\delta},N,P$ and their derivatives, but not on $\mathcal{M}$ or its derivatives. Denote
\begin{equation*}
    n^0=1,\quad n^1=\sin\theta\cos\phi,\quad n^2=\sin\theta\sin\phi,\quad n^3=\cos\theta.
\end{equation*}
The Bondi energy $m_0$ and the Bondi momentum $m_i(u)$ of each null hypersurface $u=constant$ are defined as
\begin{align*}
    m_\nu(u)=\frac{1}{4\pi}\int_{S^2}\mathcal{M}(u,\theta,\phi)n^\nu\mathrm{d}S, 
\end{align*}
where $i=1, 2, 3$ and $\nu=0, 1, 2, 3$.

Comparing the $r^{-2}$-coefficient in $\Omega_{00}=0$, we obtain
\begin{equation}\label{eq:R00eqExpand}
\begin{aligned}
    \mathcal{M}_u &= -c_u^2-d_u^2-\frac{9}{2} \tilde{I} _{uu} ^2 \\
                  &\quad +\frac{1}{2}\bigg(-6\mathcal{M}\tilde{X}\cos\theta-4\mathcal{M}_\phi\tilde{Y}-3\mathcal{M}\tilde{Y}_\phi \\
                  &\quad -4\mathcal{M}_\theta\tilde{X}\sin\theta
                     -3\mathcal{M}\tilde{X}_\theta \sin\theta+ \tensor[_{\Lambda}]{\overset{2}{\Omega}}{_{00}}\bigg) \Lambda \\
                  &\quad +\frac{1}{2}\bigg(2\mathcal{M} c\tilde{X}^2\sin^2\theta +4\mathcal{M}d\tilde{X}\tilde{Y}\sin^2\theta \\
                  &\quad -2\mathcal{M}c\tilde{Y}^2\sin^2\theta + \tensor[_{\Lambda^2}]{\overset{2}{\Omega}}{_{00}}\bigg)\Lambda^2
                     +\frac{1}{2} \tensor[_{\Lambda^3}]{\overset{2}{\Omega}}{_{00}}\Lambda^3.
\end{aligned}
\end{equation}
Therefore
\begin{align*}
    \frac{\mathrm{d}m_\nu}{\mathrm{d}u} =& -\frac{1}{8\pi}\int_{S^2}\left(2c_u^2+2d_u^2+9 \tilde{I}_{uu} ^2\right) n^\nu \mathrm{d}S \\
                   & +\frac{\Lambda}{8\pi}\int_{S^2} \Bigg[\bigg(-6\mathcal{M}\tilde{X}\cos\theta-4\mathcal{M}_\phi\tilde{Y}
-3\mathcal{M}\tilde{Y}_\phi \\
    & -4\mathcal{M}_\theta\tilde{X}\sin\theta
      -3\mathcal{M}\tilde{X}_\theta \sin\theta + \tensor[_{\Lambda}]{\overset{2}{\Omega}}{_{00}}\bigg) \Bigg] n^\nu\mathrm{d}S\\
    & +\frac{\Lambda^2}{8\pi}\int_{S^2} \bigg(2\mathcal{M} c\tilde{X}^2\sin^2\theta + 4\mathcal{M}d\tilde{X}\tilde{Y}\sin^2\theta \\
    & -2\mathcal{M}c\tilde{Y}^2\sin^2\theta + \tensor[_{\Lambda^2}]{\overset{2}{\Omega}}{_{00}}\bigg) n^\nu\mathrm{d}S\\
    & +\frac{\Lambda^3}{8\pi}\int_{S^2} \tensor[_{\Lambda^3}]{\overset{2}{\Omega}}{_{00}} n^\nu\mathrm{d}S.
\end{align*}

Note that $\mathcal{M}$ is defined on $S^2$, which takes the same value at $\phi=0$ and $\phi =2\pi$, and also at $\theta=0$ and $\theta=\pi$. We can simply the second integral for $\nu =0$.
Clearly, 
\begin{align*}
    \int_{S^2} \mathcal{M}_\phi \tilde{Y}\mathrm{d}S &= \int_0 ^\pi \int_0 ^{2\pi} \mathcal{M}_\phi \tilde{Y} \sin\theta \mathrm{d}\phi \mathrm{d}\theta \\
                                                 &=-\int_0 ^\pi \int_0 ^{2\pi} \mathcal{M} \tilde{Y} _\phi \sin\theta \mathrm{d}\phi \mathrm{d}\theta \\
                                                 &=-\int_{S^2} \mathcal{M} \tilde{Y}_\phi \mathrm{d}S,\\
    \int_{S^2} \mathcal{M}_\theta \tilde{X}\sin\theta \mathrm{d}S &= \int_0^{2\pi}\int_0^\pi  \mathcal{M}_\theta \tilde{X} \sin^2 \theta \mathrm{d}\theta \mathrm{d}\phi  \\
                                                &=-\int_0^{2\pi}\int_0^\pi \bigg(\mathcal{M}\tilde{X}_\theta \sin^2\theta +2\mathcal{M}\tilde{X} \sin \theta \cos\theta\bigg)\mathrm{d}\theta \mathrm{d}\phi \\
                                                &=-\int_{S^2} \bigg(\mathcal{M}\tilde{X}_\theta \sin \theta +2\mathcal{M}\tilde{X} \cos\theta\bigg)\mathrm{d}S.
\end{align*}
We obtain
\begin{align*}
    \frac{\mathrm{d}m_0}{\mathrm{d}u} &= -\frac{1}{8\pi}\int_{S^2}\left(2c_u^2+2d_u^2+9 \tilde{I}_{uu} ^2\right) \mathrm{d}S \\
                   &\quad +\frac{\Lambda}{8\pi}\int_{S^2} \bigg( \mathcal{M}\tilde{X}_\theta\sin\theta
    +2\mathcal{M}\tilde{X}\cos\theta +\mathcal{M}\tilde{Y}_\phi
                   + \tensor[_{\Lambda}]{\overset{2}{\Omega}}{_{00}}\bigg)\mathrm{d}S\\
    &\quad +\frac{\Lambda^2}{8\pi}\int_{S^2} \bigg(2\mathcal{M} c\tilde{X}^2\sin^2\theta + 4\mathcal{M}d\tilde{X}\tilde{Y}\sin^2\theta \\
    &\quad -2\mathcal{M}c\tilde{Y}^2\sin^2\theta + \tensor[_{\Lambda^2}]{\overset{2}{\Omega}}{_{00}}\bigg) \mathrm{d}S\\
    &\quad +\frac{\Lambda^3}{8\pi}\int_{S^2} \tensor[_{\Lambda^3}]{\overset{2}{\Omega}}{_{00}} \mathrm{d}S.
\end{align*}

\begin{remark}
According to \cite{Planck2018}, $\Lambda $ is approximately equal to $1.1 \times 10 ^{-52} m^{-2}$, and the curvatures due to $\Lambda$ is approximately equal to $10^{-18} $ times the curvatures due to gravitational waves. Thus, if $2c_u^2+2d_u^2+9 \tilde{I}_{uu} ^2>0$, then $\int_{S^2}\big(2c_u^2+2d_u^2+9 \tilde{I}_{uu} ^2\big) \mathrm{d}S$ is likely to dominant the remaining terms. In this case it should still have
\begin{align*}
\frac{\mathrm{d}m_0}{\mathrm{d}u}<0.
\end{align*}
\end{remark}

\begin{remark}
    In \cite{S1}, Saw studied the mass-loss of gravitational waves with a cosmological constant in frame of Newman-Penrose formalism. Although it is not equivalent to the Bondi-Sachs formalism studied in the current paper, we can compare formulas \textup{(95), (99)} of \cite{S1} with \eqref{eq:betaExpand}-\eqref{eq:c-dRestrictTilde} and \eqref{eq:NPbasis} of the current paper, and obtain that the asymptotic expansions in \cite{S1} corresponding to 
\begin{align}
\tilde{X}=\tilde{Y}=0\Longrightarrow c=d=0 \label{X-Y 0}
\end{align}
of the current paper. Under \eqref{X-Y 0}, we have
\begin{align*}
    \mathcal{M}_u=-\frac{9}{2}\tilde{I}_{uu}^2
+\frac{\Lambda}{2}\left( N_\theta+N\cot\theta+P_\phi\csc\theta -3\tilde{I}_u^2 \right)
+\frac{\Lambda^2}{8}\left(4\tilde{I}_uL- 27\tilde{I}_u^4\right) 
\end{align*}
and the Bondi energy-momentum satisfies
\begin{align*}
    \frac{\mathrm{d}m_0}{\mathrm{d}u} =& -\frac{9}{8\pi}\int_{S^2} \tilde{I}_{uu}^2 \mathrm{d}S 
    -\frac{3\Lambda}{8\pi}\int_{S^2} \tilde{I}_u^2 \mathrm{d}S +\frac{\Lambda^2}{32\pi}\int_{S^2} \bigg(4\tilde{I}_uL- 27\tilde{I}_u^4\bigg) \mathrm{d}S\\
    & +\frac{\Lambda}{8\pi}\int_{S^2} \bigg(  N_\theta+N\cot\theta+P_\phi\csc\theta \bigg)\mathrm{d}S\\
    =&-\frac{9}{8\pi}\int_{S^2} \tilde{I}_{uu}^2 \mathrm{d}S 
    -\frac{\Lambda}{8\pi}\int_{S^2}\bigg[3\bigg(\tilde{I}_u -\frac{\Lambda L}{6}\bigg)^2 -\frac{\Lambda ^2 L^2}{12}+\frac{27 \Lambda \tilde{I}_u^4}{4} \bigg]\mathrm{d}S\\
    &+\frac{\Lambda}{8\pi}\int_0^{2\pi}\int_0^\pi  \bigg(N_\theta \sin\theta +N\cos\theta +P_\phi \bigg) \mathrm{d}\theta \mathrm{d}\phi 
\end{align*}
Since $N$, $P$ are smooth functions defined on $S^2$, we have $P(u, \theta, 0)=P(u, \theta, 2\pi)$, thus
\begin{align*}
\int_0^{2\pi}\int_0^\pi & \bigg(N_\theta \sin\theta +N\cos\theta +P_\phi  \bigg) \mathrm{d}\theta \mathrm{d}\phi \\
=&\int_0^{2\pi}  (N\sin\theta)\bigg|_{\theta=0}^\pi \mathrm{d}\phi +\int_0^\pi P\bigg|_{\phi=0}^{2\pi} \mathrm{d}\theta =0.
\end{align*}
Therefore, if $0\leqslant \Lambda L^2 \leqslant 81 \tilde{I}_u^4$, then $\frac{\mathrm{d}m_0}{\mathrm{d}u} \leqslant 0$. 

If there is no scalar field as in \cite{S1}, then $\tilde{I}=L=0$ and $\frac{\mathrm{d}m_0}{\mathrm{d}u} = 0$ for the Bondi-Sachs formalism.
\end{remark}

\bigskip

\section*{Appendix A: Explicit Formulas for $\mathscr{R}_1,\ldots,\mathscr{R}_7$ in \eqref{eq:EQ1}-\eqref{eq:EQ7} }
\begin{align*}
    \mathscr{R}_1 &= \frac{1}{2}r\Big(\gamma_r^2\cosh^2(2\delta)+\delta_r^2\Big)+\frac{1}{4}r\Psi_r^2,  \\
    \mathscr{R}_2 &= 2r^2\Big(\beta_{r\theta}+2\delta_r\delta_\theta-2r^{-1}\beta_\theta
	            -4\gamma_r\delta_\theta\sinh(2\delta)\cosh(2\delta) \\
	           &\quad -\big(\gamma_{r\theta}-2\gamma_r\gamma_\theta+2\gamma_r\cot\theta\big)\cosh^2(2\delta) \Big)\\
	           &\quad +2r^2\mathrm{e}^{2\gamma}\csc\theta\Big(-\delta_{r\phi}
                   -2\delta_r\gamma_\phi+(\gamma_{r\phi}+2\gamma_r\gamma_\phi)\sinh(2\delta)\cosh(2\delta) \\
               &\quad +2\gamma_r\delta_\phi(1+2\sinh^2(2\delta))\Big)+2r^2\Psi_r\Psi_\theta,\\
    \mathscr{R}_3 &= 2r^2\mathrm{e}^{-2\gamma}\Big(-\delta_{r\theta}+2\delta_r\gamma_\theta-2\delta_r\cot\theta
               -(\gamma_{r\theta}-2\gamma_r\gamma_\theta\\
               &\quad +2\gamma_r\cot\theta)\cosh(2\delta)\sinh(2\delta)-2\gamma_r\delta_\theta(1+2\sinh^2(2\delta)) \Big) \\
               &\quad +2r^2\csc\theta\Big(\beta_{r\phi}+2\delta_r\delta_\phi-2r^{-1}\beta_\phi
	            +4\gamma_r\delta_\phi\cosh(2\delta)\sinh(2\delta) \\
               &\quad +(\gamma_{r\phi}+2\gamma_r\gamma_\phi)\cosh^2(2\delta) \Big)+2 r^2 \Psi_r\Psi_\phi\csc\theta,\\
    \mathscr{R}_4 &= 2\mathrm{e}^{2\beta}\csc\theta\Big((\beta_{\theta\phi}+\beta_\theta\beta_\phi+2\delta_\theta\delta_\phi)\sinh(2\delta) \\
	           &\quad +(\delta_{\theta\phi}+\delta_\phi\cot\theta+\delta_\theta\gamma_\phi-\gamma_\theta\delta_\phi
                +\delta_\theta\beta_\phi+\beta_\theta\delta_\phi)\cosh(2\delta) \Big) \\
               &\quad-\mathrm{e}^{2\beta-2\gamma}\Big( (\beta_{\theta\theta}+\beta_\theta^2+\beta_\theta\cot\theta+2\gamma_\theta^2+2\delta_\theta^2
                -1-\gamma_{\theta\theta}\\
               &\quad-3\gamma_\theta\cot\theta -2\beta_\theta\gamma_\theta)\cosh(2\delta)
                +(\delta_{\theta\theta}+3\delta_\theta\cot\theta+2\beta_\theta\delta_\theta\\
               &\quad-4\gamma_\theta\delta_\theta)\sinh(2\delta) \Big)
                -\mathrm{e}^{2\beta+2\gamma}\csc^2\theta\Big((\beta_{\phi\phi}+\beta_\phi^2+2\gamma_\phi^2+2\delta_\phi^2\\
               &\quad+\gamma_{\phi\phi}+2\beta_\phi\gamma_\phi)\cosh(2\delta)
                +(\delta_{\phi\phi}+2\beta_\phi\delta_\phi+4\gamma_\phi\delta_\phi)\sinh(2\delta) \Big) \\
	           &\quad-\frac{1}{4}r^4\mathrm{e}^{-2\beta}\Big( (\mathrm{e}^{2\gamma}U_r^2
                +\mathrm{e}^{-2\gamma}W_r^2)\cosh(2\delta)+2U_r W_r \sinh(2\delta) \Big) \\
	           &\quad +\frac{1}{2}r( rU_{r\theta}+rU_r\cot\theta+4U_\theta+4U\cot\theta)
                +\frac{1}{2}r\csc\theta (r W_{r\phi}+4W_\phi) \\
               &\quad -\frac{1}{2}\mathrm{e}^{2\beta}\Big((\mathrm{e}^{-2\gamma}\Psi_\theta^2+\mathrm{e}^{2\gamma}\Psi_\phi^2\csc^2\theta)\cosh(2\delta)
                -2\Psi_\theta\Psi_\phi \sinh(2\delta)\csc\theta\Big) \\
               &\quad -\Lambda \mathrm{e}^{2\beta}r^2,\\
    \mathscr{R}_5 &= \frac{1}{2}(\gamma_r V_r+\gamma_{rr}V+r^{-1}\gamma_r V)
	             \cosh(2\delta)+2\gamma_r\delta_r V\sinh(2\delta) \\
	            &\quad+\frac{1}{8}r^3\mathrm{e}^{-2\beta}(\mathrm{e}^{2\gamma}U_r^2-\mathrm{e}^{-2\gamma}W_r^2)
	             +\frac{1}{2}r^{-1}\mathrm{e}^{2\beta-2\gamma}(\beta_{\theta\theta}+\beta_\theta^2-\beta_\theta\cot\theta) \\
	            &\quad-\frac{1}{2}r^{-1}\mathrm{e}^{2\beta+2\gamma}(\beta_{\phi\phi}+\beta_\phi^2)\csc^2\theta
	             +r^{-1}\mathrm{e}^{2\beta}(\beta_\theta\delta_\phi-\beta_\phi\delta_\theta)\csc\theta \\
	            &\quad+\frac{1}{4}r\mathrm{e}^{2\gamma}\csc\theta
                  \Big( (U_{r\phi}+2r^{-1}U_\phi)\sinh(2\delta)+4\delta_r U_\phi\cosh(2\delta) \Big) \\
                &\quad-\frac{1}{4}r\mathrm{e}^{-2\gamma}\Big( (W_{r\theta}-W_r\cot\theta)\sinh(2\delta)
                 +2r^{-1}(W_\theta-W\cot\theta)\sinh(2\delta) \\
	            &\quad+4\delta_r(W_\theta-W\cot\theta)\cosh(2\delta) \Big)
	             -\frac{1}{4}r\big( U_{r\theta}+2r^{-1}U_\theta-U_r\cot\theta\\
                &\quad-2r^{-1}U\cot\theta+4r^{-1}\gamma_\theta U+4\gamma_{r\theta}U+2\gamma_\theta U_r+2\gamma_rU_\theta\\
                &\quad+2\gamma_r U\cot\theta \big)\cosh(2\delta)
	             -r(\delta_r U_\theta+2\gamma_r\delta_\theta U+2\delta_r\gamma_\theta U\\
                &\quad-\delta_r U\cot\theta)\sinh(2\delta)
	             +\frac{1}{4}r\csc\theta( W_{r\phi}+2r^{-1}W_\phi-4r^{-1}\gamma_\phi W\\
                &\quad-4\gamma_{r\phi}W-2\gamma_\phi W_r-2\gamma_r W_\phi)\cosh(2\delta)+r\csc\theta(\delta_r W_\phi
                 -2\delta_r\gamma_\phi W\\
                &\quad-2\gamma_r\delta_\phi W)\sinh(2\delta)
                 +\frac{1}{4r}\mathrm{e}^{2\beta}\left(\mathrm{e}^{-2\gamma}\Psi_\theta^2-\mathrm{e}^{2\gamma}\Psi_\phi^2\csc^2\theta \right),\\
    \mathscr{R}_6 &=\frac{1}{2}\Big(\delta_r V_r+\delta_{rr}V+r^{-1}\delta_r V-2\gamma_r^2 V\cosh(2\delta)\sinh(2\delta)\Big)
	             -\frac{1}{2r}\mathrm{e}^{2\beta-2\gamma}\\
                &\quad \cdot(\beta_{\theta\theta}+\beta_\theta^2-\beta_\theta\cot\theta)\sinh(2\delta)
                 -\frac{1}{2r}\mathrm{e}^{2\beta+2\gamma}\csc^2\theta(\beta_{\phi\phi}+\beta_\phi^2)\sinh(2\delta) \\
	            &\quad -\frac{1}{r}\mathrm{e}^{2\beta}\csc\theta\big( -\beta_{\theta\phi}-\beta_\theta\beta_\phi+\beta_\phi\cot\theta
                 +\beta_\theta\gamma_\phi-\gamma_\theta\beta_\phi \big)\cosh(2\delta) \\
	            &\quad+\frac{1}{8}r^3\mathrm{e}^{-2\beta}\Big((\mathrm{e}^{2\gamma}U_r^2+\mathrm{e}^{-2\gamma}W_r^2)\sinh(2\delta)
                 +2U_r W_r\cosh(2\delta) \Big) \\
                &\quad-\frac{1}{2}r\Big( 2\delta_{r\theta}U+\frac{2}{r}\delta_\theta U+\delta_r U_\theta+\delta_\theta U_r+\delta_r U\cot\theta
                 -2(\gamma_r U_\theta\\
                &\quad-\gamma_r U\cot\theta+2\gamma_r\gamma_\theta U)\cosh(2\delta)\sinh(2\delta) \Big)
                 -\frac{1}{2}r\csc\theta\Big( 2\delta_{r\phi}W\\
                &\quad+\frac{2}{r}\delta_\phi W +\delta_r W_\phi+\delta_\phi W_r +2(\gamma_r W_\phi-2\gamma_r\gamma_\phi W)\cosh(2\delta)\sinh(2\delta)\Big) \\
                &\quad -\frac{1}{4}r\mathrm{e}^{-2\gamma}\Big( W_{r\theta}-W_r\cot\theta+ \frac{2}{r}(W_\theta-W\cot\theta)-4\gamma_r(W_\theta-W\cot\theta)\\
                &\quad\cdot\cosh^2(2\delta)\Big)-\frac{1}{4}r\mathrm{e}^{2\gamma}\csc\theta
	             \Big(U_{r\phi}+\frac{2}{r}U_\phi+4\gamma_rU_\phi\cosh^2(2\delta) \Big) \\
                &\quad-\frac{\mathrm{e}^{2\beta}}{4r}\Big((\mathrm{e}^{-2\gamma}\Psi_\theta^2+\mathrm{e}^{2\gamma}\Psi_\phi^2\csc^2\theta)\sinh(2\delta)
                 -2\Psi_\theta\Psi_\phi\cosh(2\delta)\csc\theta \Big),\\
\mathscr{R}_7 &=\frac{1}{2}V\Psi_{rr}-rU\Psi_{r\theta}-rW\csc\theta\Psi_{r\phi}
                 +\frac{1}{2r}\mathrm{e}^{2\beta-2\gamma}\cosh(2\delta)\Psi_{\theta\theta} \\
                &\quad-\frac{1}{r}\mathrm{e}^{2\beta}\Psi_{\theta\phi}\csc\theta\sinh(2\delta)
                 +\frac{1}{2r}\mathrm{e}^{2\beta+2\gamma}\Psi_{\phi\phi}\cosh(2\delta)\csc^2\theta \\
	            &\quad+\frac{r}{2}\Big(\frac{V}{r^2}+\frac{V_r}{r}-U\cot\theta-U_\theta-W_\phi\csc\theta\Big)\Psi_r
                 +\bigg( \frac{1}{2r}\mathrm{e}^{2\beta-2\gamma}\\
                &\quad\cdot\Big(\cosh(2\delta)\cot\theta
                 +2\cosh(2\delta)\beta_\theta+2\sinh(2\delta)\delta_\theta-2\cosh(2\delta)\gamma_\theta\Big)\\
                &\quad-\frac{1}{2}(2U+rU_r)
                 -\frac{1}{r}\mathrm{e}^{2\beta}\cosh(2\delta)\csc\theta\delta_\phi 
                 -\frac{1}{r}\mathrm{e}^{2\beta}\sinh(2\delta)\csc\theta\beta_\phi \bigg)\Psi_\theta\\
                &\quad+\bigg(\frac{1}{2r} \mathrm{e}^{2\beta+2\gamma}\csc^2\theta\Big( 2\cosh(2\delta)\gamma_\phi
                 +2\sinh(2\delta)\delta_\phi+2\cosh(2\delta)\beta_\phi\Big)\\
                &\quad-\frac{1}{2}(rW_r+2W)\csc\theta
                 -\frac{1}{r}\mathrm{e}^{2\beta}\beta_\theta\sinh(2\delta)\csc\theta\\
                &\quad-\frac{1}{r}\mathrm{e}^{2\beta}\delta_\theta\cosh(2\delta)\csc\theta \bigg)\Psi_\phi.
\end{align*}

\section*{Appendix B: Explicit Higher-Order Coefficients}

The following functions are some explicit higher-order coefficients in asymptotic expansions in Section 3.

\begin{align*}
    \overset{4}{\beta} &= \frac{1}{8}\bigg( ( c^2+d^2)^2-6(cC+dD)-9\tilde{I}_u L \bigg),\\
    \overset{4}{U} &=\frac{1}{8}\bigg( -16 c^3\cot\theta-4c^2(3 c_\theta+2 d_\phi \csc\theta)
    -8 d^2 c_\theta \\
    &\quad 
    -2c(2dc_\phi \csc\theta 
    +8d^2\cot\theta
    +2d d_\theta
    +12N 
    -9\tilde{I}_u^2\cot\theta)  \\
    &\quad 
    +12 C_\theta 
        +24C\cot\theta 
        -12d^2 d_\phi\csc\theta
        -24 dP 
        +12 D_\phi \csc\theta \\
        &\quad 
        +9\tilde{I}_u^2 c_\theta  
        +9 \tilde{I}_u^2 d_\phi \csc\theta
    \bigg),  \\
    \overset{5}{U} &=\frac{1}{20}\bigg( 16 c^4\cot\theta - 64 cC\cot\theta + 32 c^2d^2\cot\theta + 16 d^4\cot\theta - 64 dD\cot\theta \\
    &\quad + 36 c^2N + 36 d^2N  - 8 c^2d c_\phi\csc\theta - 8 d^3c_\phi\csc\theta + 32 dC_\phi\csc\theta  \\
    &\quad + 8 c^3d_\phi\csc\theta + 8 cd^2d_\phi\csc\theta- 32 cD_\phi\csc\theta
    + 20 c^3c_\theta + 10 Cc_\theta \\
    &\quad + 20 cd^2c_\theta  - 38 cC_\theta + 20 c^2d d_\theta + 20 d^3d_\theta + 10 Dd_\theta \\
    &\quad - 38 dD_\theta + 32 \overset{4}{\gamma}\cot\theta + 16 \overset{4}{\gamma}_{\theta}+ 16 \overset{4}{\delta}_{\phi}\csc\theta  + 15 \tilde{I}_uL_\theta \\
    &\quad  - 72 c^2\tilde{I}_u^2\cot\theta - 72 d^2\tilde{I}_u^2\cot\theta  - 54 N\tilde{I}_u^2+ 36 d\tilde{I}_u^2c_\phi\csc\theta  - 36 c\tilde{I}_u^2d_\phi\csc\theta  \\
    &\quad - 36 c\tilde{I}_u^2c_\theta  - 36 d\tilde{I}_u^2d_\theta  \bigg), \\
    \overset{4}{W} &= \frac{1}{8}\bigg(
        4 c^2 (3c_\phi\csc\theta 
        -2 d_\theta  
        -4 d\cot\theta)
        - 4d^2 (3d_\theta
        -2 c_\phi \csc\theta)  \\
    &\quad 
        -4c(d c_\theta -d d_\phi\csc\theta 
        -6 P) - 12 (C_\phi\csc\theta 
        -D_\theta-2 D\cot\theta) \\
    &\quad 
         -16d^3\cot\theta -24 d N 
        +9\tilde{I}_u^2 (d_\theta 
        +2d \cot\theta
        -c_\phi \csc\theta)
    \bigg),  \\
    \overset{5}{W} &= \frac{1}{20}\bigg(
        -96 C d\cot\theta
 +96 cD\cot\theta
 +36 c^2P
 +36 d^2P
 +20 c^3c_{\phi}\csc\theta \\
&\quad +10 Cc_{\phi}\csc\theta
 +20 cd^2c_{\phi}\csc\theta
 -38 cC_{\phi}\csc\theta
 +20 c^2dd_{\phi}\csc\theta
 +20 d^3d_{\phi}\csc\theta \\
&\quad +10 Dd_{\phi}\csc\theta
 -38 dD_{\phi}\csc\theta
 +8c^2dc_{\theta}
 +8d^3c_{\theta}
 -32dC_{\theta} \\
&\quad -8c^3d_{\theta}
 -8cd^2d_{\theta}
 +32cD_{\theta}
 +32\overset{4}{\delta}\cot\theta
 -16\overset{4}{\gamma}_{\phi}\csc\theta \\
&\quad +16\overset{4}{\delta}_{\theta}
 -54P\tilde{I}_u^2
 -36c\tilde{I}_u^2c_{\phi}\csc\theta
 -36d\tilde{I}_u^2d_{\phi}\csc\theta
 +15\tilde{I}_uL_{\phi}\csc\theta \\
&\quad -36d\tilde{I}_u^2c_{\theta}
 +36c\tilde{I}_u^2d_{\theta}
    \bigg),  \\
    \overset{1}{V} &= \frac{1}{96}\bigg[ 8\Big(2(25\cos^2\theta-1) c^2\csc^2\theta +2c(5c_{\phi\phi}\csc^2\theta+37c_\theta\cot\theta\\
                &\quad +5c_{\theta\theta}+24 d_\phi\cot\theta\csc\theta)+2d(-24 c_\phi \cot\theta\csc\theta +37 d_\theta\cot\theta \\
                &\quad +5d_{\theta\theta}+5d_{\phi\phi}\csc^2\theta )-32c_\phi d_\theta\csc\theta +32 c_\theta d_\phi\csc\theta +22 c_\theta^2\\
                &\quad +22c_\phi^2\csc^2\theta +2(24\csc^2\theta-25)d^2+22d_\phi^2\csc^2\theta+22 d_\theta^2\\
                &\quad -12 N_\theta-12 N\cot\theta-12 P_\phi\csc\theta +27\tilde{I}_u^2 \Big) \\
    &\quad +9\Lambda \Big(
    4(c^2+d^2)^2 -16(cC+dD)
        +27\tilde{I}_u^4
        +12(c^2+d^2)\tilde{I}_u^2
        -24\tilde{I}_u L
    \Big) \bigg], \\
  \Theta &= 2r_{-1}^4 \theta_1\cos\hat{\theta}\sin\hat{\theta} +2r_{-1}^4 \phi_{1,\hat{\phi}}\sin^2\hat{\theta} -2r_{-1}^4\theta_1 Yu_{0,\hat{\phi}}\cos\hat{\theta}  \\
  &\quad -4\mathrm{e}^{2B}r_{-1}^3 B_\phi u_{0,\hat{\phi}} -2r_{-1}^4 Y\phi_{1,\hat{\phi}} u_{0,\hat{\phi}}\sin\hat{\theta} -2\mathrm{e}^{2B} r_{-1}^2 r_{-1,\hat{\phi}} u_{0,\hat{\phi}}  \\
  &\quad -\mathrm{e}^{2B}r_{-1}^3 X u_{0,\hat{\phi}}^2\cot\hat{\theta} +4\mathrm{e}^{2B} r_{-1}^3 Y B_\phi u_{0,\hat{\phi}}^2\csc\hat{\theta} -2r_{-1}^4 Y u_{1,\hat{\phi}}\sin\hat{\theta}  \\
  &\quad +2r_{-1}^4 Y^2 u_{0,\hat{\phi}} u_{1,\hat{\phi}} -2\phi_1 r_{-1}^4 u_{0,\hat{\phi}} Y_\phi \sin\hat{\theta} -\mathrm{e}^{2B}r_{-1}^3 u_{0,\hat{\phi}}^2 Y_\phi \csc\hat{\theta}  \\
  &\quad +2\phi_1 r_{-1}^4 Y u_{0,\hat{\phi}}^2 Y_\phi +2r_{-1}^4 X u_{0,\hat{\phi}}\phi_{1,\hat{\theta}} \sin^2\hat{\theta} -2r_{-1}^4 X Y u_{0,\hat{\phi}}^2 \phi_{1,\hat{\theta}}\sin\hat{\theta}  \\
  &\quad -2\mathrm{e}^{2B} r_{-1}^2 X u_{0,\hat{\phi}}^2 r_{-1,\hat{\theta}} +2r_{-1}^4 \theta_{1,\hat{\theta}}\sin^2\hat{\theta} -4 r_{-1}^4 Y u_{0,\hat{\phi}} \theta_{1,\hat{\theta}}\sin\hat{\theta}  \\
  &\quad +2r_{-1}^4 Y^2 u_{0,\hat{\phi}}^2 \theta_{1,\hat{\theta}} -4 r_{-1}^4 \theta_1 X u_{0,\hat{\theta}} \cos\hat{\theta} \sin\hat{\theta}  \\
  &\quad -4r_{-1}^4 X\phi_{1,\hat{\phi}} u_{0,\hat{\theta}}\sin^2\hat{\theta} +2r_{-1}^4 Y \theta_{1,\hat{\phi}} u_{0,\hat{\theta}} \sin\hat{\theta}  \\
  &\quad +2r_{-1}^4 \theta_1 XYu_{0,\hat{\phi}} u_{0,\hat{\theta}} \cos\hat{\theta} +4\mathrm{e}^{2B}r_{-1}^3 XB_\phi u_{0,\hat{\phi}}u_{0,\hat{\theta}}  \\
  &\quad +2r_{-1}^4 XY\phi_{1,\hat{\phi}} u_{0,\hat{\phi}} u_{0,\hat{\theta}}\sin\hat{\theta} +2\mathrm{e}^{2B} r_{-1}^2 X r_{-1,\hat{\phi}} u_{0,\hat{\phi}} u_{0,\hat{\theta}}  \\
  &\quad -2r_{-1}^4 Y^2 \theta_{1,\hat{\phi}} u_{0,\hat{\phi}} u_{0,\hat{\theta}} +2 r_{-1}^4 XYu_{1,\hat{\phi}} u_{0,\hat{\theta}}\sin\hat{\theta} -2\phi_1 r_{-1}^4 X_\phi u_{0,\hat{\theta}}\sin^2\hat{\theta}  \\
  &\quad +2\phi_1 r_{-1}^4 Y u_{0,\hat{\phi}} X_\phi u_{0,\hat{\theta}}\sin\hat{\theta} +2\phi_1 r_{-1}^4 X u_{0,\hat{\phi}} Y_\phi u_{0,\hat{\theta}}\sin\hat{\theta}  \\
  &\quad -4\mathrm{e}^{2B} r_{-1}^3 B_\theta u_{0,\hat{\theta}}\sin^2\hat{\theta} +4\mathrm{e}^{2B}r_{-1}^3 Y u_{0,\hat{\phi}}B_\theta u_{0,\hat{\theta}}\sin\hat{\theta}  \\
  &\quad -2r_{-1}^4 X^2 u_{0,\hat{\phi}} \phi_{1,\hat{\theta}} u_{0,\hat{\theta}}\sin^2\hat{\theta} -2\mathrm{e}^{2B} r_{-1}^2 r_{-1,\hat{\theta}}u_{0,\hat{\theta}}\sin^2\hat{\theta}  \\
  &\quad +2\mathrm{e}^{2B} r_{-1}^2 Yu_{0,\hat{\phi}} r_{-1,\hat{\theta}}u_{0,\hat{\theta}}\sin\hat{\theta} -2r_{-1}^4 X\theta_{1,\hat{\theta}} u_{0,\hat{\theta}}\sin^2\hat{\theta}  \\
  &\quad +2r_{-1}^4 XYu_{0,\hat{\phi}}\theta_{1,\hat{\theta}} u_{0,\hat{\theta}}\sin\hat{\theta} -\mathrm{e}^{2B}r_{-1}^3 X u_{0,\hat{\theta}}^2 \cos\hat{\theta}\sin\hat{\theta}  \\
  &\quad +2r_{-1}^4 \theta_1 X^2u_{0,\hat{\theta}}^2 \cos\hat{\theta}\sin\hat{\theta} +2r_{-1}^4 X^2\phi_{1,\hat{\phi}} u_{0,\hat{\theta}}^2\sin^2\hat{\theta}  \\
  &\quad -2\mathrm{e}^{2B} r_{-1}^2 Y r_{-1,\hat{\phi}} u_{0,\hat{\theta}}^2\sin\hat{\theta} -2r_{-1}^4 X Y\theta_{1,\hat{\phi}} u_{0,\hat{\theta}}^2\sin\hat{\theta}  \\
  &\quad +2\phi_1 r_{-1}^4 XX_\phi u_{0,\hat{\theta}}^2 \sin^2\hat{\theta} -\mathrm{e}^{2B} r_{-1}^3 Y_\phi u_{0,\hat{\theta}}^2 \sin\hat{\theta}  \\
  &\quad +4\mathrm{e}^{2B}r_{-1}^3 X B_\theta u_{0,\hat{\theta}}^2\sin^2\hat{\theta} -2r_{-1}^4 X u_{1,\hat{\theta}}\sin^2\hat{\theta} +2r_{-1}^4 XY u_{0,\hat{\phi}} u_{1,\hat{\theta}}\sin\hat{\theta}  \\
  &\quad +2r_{-1}^4 X^2 u_{0,\hat{\theta}} u_{1,\hat{\theta}}\sin^2\hat{\theta} -\mathrm{e}^{2B} r_{-1}^3 u_{0,\hat{\phi}}^2 X_\theta -2r_{-1}^4 \theta_1 u_{0,\hat{\theta}} X_\theta \sin^2\hat{\theta}  \\
  &\quad +2r_{-1}^4 \theta_1 Y u_{0,\hat{\phi}} u_{0,\hat{\theta}}X_\theta \sin\hat{\theta} -\mathrm{e}^{2B} r_{-1}^3 u_{0,\hat{\theta}}^2 X_\theta \sin^2\hat{\theta}  \\
  &\quad +2r_{-1}^4 \theta_1 X u_{0,\hat{\theta}}^2 X_\theta\sin^2\hat{\theta} -2r_{-1}^4 \theta_1 u_{0,\hat{\phi}} Y_\theta\sin\hat{\theta}  \\
  &\quad +2r_{-1}^4 \theta_1 Y u_{0,\hat{\phi}}^2 Y_\theta +2r_{-1}^4 \theta_1 X u_{0,\hat{\phi}} u_{0,\hat{\theta}} Y_\theta \sin\hat{\theta} -2r_{-1}^4 u_1 u_{0,\hat{\theta}} X_u\sin^2\hat{\theta}  \\
  &\quad +2r_{-1}^4 u_1 Y u_{0,\hat{\phi}} u_{0,\hat{\theta}} X_u\sin\hat{\theta} +2r_{-1}^4 u_1 X u_{0,\hat{\theta}}^2 X_u\sin^2\hat{\theta}  \\
  &\quad -2r_{-1}^4 u_1 u_{0,\hat{\phi}} Y_u\sin\hat{\theta} +2r_{-1}^4 u_1 Y u_{0,\hat{\phi}}^2 Y_u +2r_{-1}^4 u_1 X u_{0,\hat{\phi}} u_{0,\hat{\theta}} Y_u \sin\hat{\theta}  \\
  &\quad +\mathrm{e}^{4B}\Lambda\bigg( \frac{2}{3} cr_{-1}^3 u_{0,\hat{\phi}}^2 +\frac{4}{3} \phi_1 r_{-1}^4 B_\phi u_{0,\hat{\phi}}^2 + \frac{2}{3} r_{-1}^4 u_{0,\hat{\phi}} u_{1,\hat{\phi}}  \\
  &\qquad +\frac{4}{3} r_{-1}^4 \theta_1 u_{0,\hat{\phi}}^2 B_\theta +\frac{2}{3} r_{-1}^4 u_{0,\hat{\phi}}^2 \theta_{1,\hat{\theta}} -\frac{4}{3} d r_{-1}^3 u_{0,\hat{\phi}} u_{0,\hat{\theta}}\sin\hat{\theta}  \\
  &\qquad -\frac{2}{3} r_{-1}^4 \theta_{1,\hat{\phi}} u_{0,\hat{\phi}} u_{0,\hat{\theta}} -\frac{2}{3} r_{-1}^4 u_{0,\hat{\phi}}\phi_{1,\hat{\theta}} u_{0,\hat{\theta}}\sin^2\hat{\theta} -\frac{2}{3} cr_{-1}^3 u_{0,\hat{\theta}}^2 \sin^2\hat{\theta}  \\
  &\qquad +\frac{2}{3} r_{-1}^4 \theta_1 u_{0,\hat{\theta}}^2 \cos\hat{\theta}\sin\hat{\theta} +\frac{4}{3} \phi_1 r_{-1}^4 B_\phi u_{0,\hat{\theta}}^2\sin^2\hat{\theta} +\frac{2}{3} r_{-1}^4 \phi_{1,\hat{\phi}} u_{0,\hat{\theta}}^2\sin^2\hat{\theta}  \\
  &\qquad +\frac{4}{3} r_{-1}^4 \theta_1 B_\theta u_{0,\hat{\theta}}^2 \sin^2\hat{\theta} +\frac{2}{3} r_{-1}^4 u_{0,\hat{\theta}} u_{1,\hat{\theta}}\sin^2\hat{\theta}  \\
  &\qquad +\frac{4}{3} r_{-1}^4 u_1 u_{0,\hat{\phi}}^2 B_u +\frac{4}{3} r_{-1}^4 u_1 u_{0,\hat{\theta}}^2 B_u\sin^2\hat{\theta} \bigg)
\end{align*}

\section*{Appendix C: Explicit Formulas for $\tensor[_{\Lambda}]{\overset{2}{\Omega}}{_{00}}$, $\tensor[_{\Lambda^2}]{\overset{2}{\Omega}}{_{00}}$, $\tensor[_{\Lambda^3}]{\overset{2}{\Omega}}{_{00}}$}

\begin{align*}
 \tensor[_{\Lambda}]{\overset{2}{\Omega}}{_{00}} &= \frac{7 c^2}{72} + \frac{7 c^2\cot^2\theta}{72} - \frac{103 c^2\csc^2\theta}{72} + \frac{7 d^2}{72} + \frac{7 d^2\cot^2\theta}{72} + N\cot\theta\\
 &\quad - \frac{103 d^2\csc^2\theta}{72}  - \frac{d\tilde{Y}\cos\theta}{2} - \frac{d\tilde{Y}\cos\theta\cot^2\theta}{2} + \frac{d\tilde{Y}\cot\theta\csc\theta}{2} \\
 &\quad - 3 \tilde{I}_{u}^2 - 18 \tilde{X}\tilde{I}_{u}\tilde{I}_{uu}\cos\theta + \frac{4 dc_{\phi}\cot\theta\csc\theta}{3} + 5 \tilde{Y}c_{\phi} + 5 \tilde{Y}c_{\phi}\cot^2\theta \\
 &\quad - \tilde{Y}c_{\phi}\csc^2\theta + \frac{c_{\phi}^2\csc^2\theta}{2} - \frac{4 cd_{\phi}\cot\theta\csc\theta}{3} - \tilde{X}d_{\phi} - \tilde{X}d_{\phi}\cot^2\theta \\
 &\quad + 5 \tilde{X}d_{\phi}\csc^2\theta + \frac{d_{\phi}^2\csc^2\theta}{2} + P_{\phi}\csc\theta + d\tilde{X}_{\phi} + 3 d\tilde{X}_{\phi}\cot^2\theta \\
 &\quad + d\tilde{X}_{\phi}\csc^2\theta + 2 c_{\phi}\tilde{X}_{\phi}\cot\theta\csc\theta + 2 c\tilde{Y}_{\phi} - 9 \tilde{I}_{u}\tilde{I}_{uu}\tilde{Y}_{\phi}  \\ 
 &\quad + \frac{5 cc_{\phi\phi}\csc^2\theta}{6} - \tilde{Y}_{\phi}c_{\phi\phi}\csc^2\theta + \frac{5 dd_{\phi\phi}\csc^2\theta}{6} + \tilde{X}_{\phi}d_{\phi\phi}\csc^2\theta  \\
 &\quad + d_{\phi}\tilde{X}_{\phi\phi}\csc^2\theta + 4 d\tilde{Y}_{\phi\phi}\cot\theta\csc\theta - c_{\phi}\tilde{Y}_{\phi\phi}\csc^2\theta + \tilde{Y}c_{\phi\phi\phi}\csc^2\theta  \\
 &\quad + c\tilde{Y}_{\phi\phi\phi}\csc^2\theta + \frac{5 cc_{\theta}\cot\theta}{6} - 4 \tilde{X}c_{\theta}\cos\theta\cot\theta + 7 \tilde{X}c_{\theta}\csc\theta  \\ 
         &\quad + 2 d_{\phi}\tilde{Y}_{\phi}\cot\theta\csc\theta
         + 4 c\tilde{X}_{\phi\phi}\cot\theta\csc\theta
         + \tilde{X}d_{\phi\phi\phi}\csc^2\theta
         - 2 \tilde{X}c_{\theta}\sin\theta \\  
 &\quad - 3 \tilde{Y}_{\phi}c_{\theta}\cot\theta + \tilde{X}_{\phi\phi}c_{\theta}\csc\theta + \frac{c_{\theta}^2}{2} + \frac{5 dd_{\theta}\cot\theta}{6} + \tilde{Y}d_{\theta}\cos\theta\cot\theta \\
 &\quad - 4 \tilde{Y}d_{\theta}\csc\theta - \tilde{Y}d_{\theta}\sin\theta + 3 \tilde{X}_{\phi}d_{\theta}\cot\theta + \tilde{Y}_{\phi\phi}d_{\theta}\csc\theta + \frac{d_{\theta}^2}{2} \\
 &\quad + N_{\theta} + 5 c\tilde{X}_{\theta}\cos\theta\cot\theta + 3 c\tilde{X}_{\theta}\sin\theta - 9 \tilde{I}_{u}\tilde{I}_{uu}\tilde{X}_{\theta}\sin\theta  \\ 
 &\quad + c_{\phi\phi}\tilde{X}_{\theta}\csc\theta + 4 c_{\theta}\tilde{X}_{\theta}\cos\theta + 2 d\tilde{Y}_{\theta}\cos\theta\cot\theta - 2 d\tilde{Y}_{\theta}\sin\theta  \\
 &\quad + d_{\phi\phi}\tilde{Y}_{\theta}\csc\theta + 4 d_{\theta}\tilde{Y}_{\theta}\cos\theta - 3 \tilde{Y}c_{\theta\phi}\cot\theta - 3 \tilde{X}d_{\theta\phi}\cot\theta  \\
 &\quad + c\tilde{Y}_{\theta\phi}\cot\theta + 3 \tilde{X}c_{\theta\phi\phi}\csc\theta - 3 \tilde{Y}d_{\theta\phi\phi}\csc\theta + c\tilde{X}_{\theta\phi\phi}\csc\theta  \\ 
       &\quad - d_{\phi}\tilde{X}_{\theta}\cot\theta
       + c_{\phi}\tilde{Y}_{\theta}\cot\theta
       - 4 d\tilde{X}_{\theta\phi}\cot\theta
       - 2 d\tilde{Y}_{\theta\phi\phi}\csc\theta  \\ 
 &\quad + \frac{5 cc_{\theta\theta}}{6} - 3 \tilde{X}c_{\theta\theta}\cos\theta - \tilde{Y}_{\phi}c_{\theta\theta} + \tilde{X}_{\theta}c_{\theta\theta}\sin\theta + \frac{5 dd_{\theta\theta}}{6} \\
 &\quad + 3 \tilde{Y}d_{\theta\theta}\cos\theta + \tilde{X}_{\phi}d_{\theta\theta} + \tilde{Y}_{\theta}d_{\theta\theta}\sin\theta - c\tilde{X}_{\theta\theta}\cos\theta + d_{\phi}\tilde{X}_{\theta\theta} \\
 &\quad + c_{\theta}\tilde{X}_{\theta\theta}\sin\theta + 2 d\tilde{Y}_{\theta\theta}\cos\theta - c_{\phi}\tilde{Y}_{\theta\theta} + d_{\theta}\tilde{Y}_{\theta\theta}\sin\theta - 3 \tilde{Y}c_{\theta\theta\phi} \\
 &\quad - 3 \tilde{X}d_{\theta\theta\phi} - 2 d\tilde{X}_{\theta\theta\phi} - c\tilde{Y}_{\theta\theta\phi} - \tilde{X}c_{\theta\theta\theta}\sin\theta + \tilde{Y}d_{\theta\theta\theta}\sin\theta \\
 &\quad - c\tilde{X}_{\theta\theta\theta}\sin\theta - 4 c\tilde{X}c_{u}\cos\theta - 8 d\tilde{Y}c_{u}\cos\theta - \tilde{Y}c_{\phi}c_{u} + 4 \tilde{X}d_{\phi}c_{u} \\
 &\quad + 4 d\tilde{X}_{\phi}c_{u} - 2 c\tilde{Y}_{\phi}c_{u} - \tilde{X}c_{\theta}c_{u}\sin\theta - 4 \tilde{Y}d_{\theta}c_{u}\sin\theta - 2 c\tilde{X}_{\theta}c_{u}\sin\theta \\
 &\quad - 4 d\tilde{Y}_{\theta}c_{u}\sin\theta - 4 d\tilde{X}d_{u}\cos\theta + 8 c\tilde{Y}d_{u}\cos\theta - 4 \tilde{X}c_{\phi}d_{u} - \tilde{Y}d_{\phi}d_{u} \\
 &\quad - 4 c\tilde{X}_{\phi}d_{u} - 2 d\tilde{Y}_{\phi}d_{u} + 4 \tilde{Y}c_{\theta}d_{u}\sin\theta - \tilde{X}d_{\theta}d_{u}\sin\theta - 2 d\tilde{X}_{\theta}d_{u}\sin\theta \\
 &\quad + 4 c\tilde{Y}_{\theta}d_{u}\sin\theta - 6 \tilde{X}N_{u}\sin\theta - 6 \tilde{Y}P_{u}\sin\theta + 4 d\tilde{X}c_{u\phi} - c\tilde{Y}c_{u\phi} \\
 &\quad - 4 c\tilde{X}d_{u\phi} - d\tilde{Y}d_{u\phi} - c\tilde{X}c_{u\theta}\sin\theta - 4 d\tilde{Y}c_{u\theta}\sin\theta - d\tilde{X}d_{u\theta}\sin\theta \\
 &\quad + 4 c\tilde{Y}d_{u\theta}\sin\theta, \\
           \tensor[_{\Lambda^2}]{\overset{2}{\Omega}}{_{00}} &=
      - \frac{4 c^4}{9} + \frac{2 cC}{3} - \frac{8 c^2d^2}{9} - \frac{4 d^4}{9} + \frac{2 dD}{3} - \frac{8 c^3\tilde{X}\cos\theta}{9} \\
 &\quad + 4 C\tilde{X}\cos\theta - \frac{8 cd^2\tilde{X}\cos\theta}{9} - \frac{8 cN\tilde{X}\sin\theta}{3} - \frac{8 dP\tilde{X}\sin\theta}{3} \\
 &\quad - 2 c^2\tilde{X}^2 - \frac{4 c^2\tilde{X}^2\cos 2\theta}{3} - 2 d^2\tilde{X}^2 - \frac{4 d^2\tilde{X}^2\cos 2\theta}{3} \\
 &\quad - 17 N\tilde{X}^2\cos\theta\sin\theta - \frac{8 c^2d\tilde{Y}\cos\theta}{9} - \frac{8 d^3\tilde{Y}\cos\theta}{9} \\
 &\quad + 4 D\tilde{Y}\cos\theta - \frac{8 dN\tilde{Y}\sin\theta}{3} + \frac{8 cP\tilde{Y}\sin\theta}{3} - 8 P\tilde{X}\tilde{Y}\sin 2\theta \\
 &\quad + 2 c^2\tilde{Y}^2 - \frac{2 c^2\tilde{Y}^2\cos 2\theta}{3} + 2 d^2\tilde{Y}^2 - \frac{2 d^2\tilde{Y}^2\cos 2\theta}{3} \\
 &\quad - N\tilde{Y}^2\cos\theta\sin\theta + L\tilde{I}_{u} - 3 c^2\tilde{I}_{u}^2 - 3 d^2\tilde{I}_{u}^2 - 9 c\tilde{X}\tilde{I}_{u}^2\cos\theta \\
 &\quad - 9 \tilde{X}^2\tilde{I}_{u}^2\cos^2\theta + 9 \tilde{X}^2\tilde{I}_{u}^2\sin^2\theta - 9 d\tilde{Y}\tilde{I}_{u}^2\cos\theta - \frac{27 \tilde{I}_{u}^4}{4} \\
 &\quad - \frac{8 cd\tilde{X}c_{\phi}}{3} - \frac{4 d\tilde{X}^2c_{\phi}\cos\theta}{3} + \frac{28 c^2\tilde{Y}c_{\phi}}{9} + \frac{4 d^2\tilde{Y}c_{\phi}}{9} \\
 &\quad + \frac{4 c\tilde{X}\tilde{Y}c_{\phi}\cos\theta}{3} - \frac{20 d\tilde{Y}^2c_{\phi}\cos\theta}{3} + \frac{9 \tilde{Y}\tilde{I}_{u}^2c_{\phi}}{2} - \frac{5 \tilde{X}^2c_{\phi}^2}{6} \\
 &\quad + \frac{11 \tilde{Y}^2c_{\phi}^2}{6} - 2 \tilde{Y}C_{\phi} - \frac{4 c^2\tilde{X}d_{\phi}}{9} - \frac{28 d^2\tilde{X}d_{\phi}}{9} + \frac{4 c\tilde{X}^2d_{\phi}\cos\theta}{3} \\
 &\quad + \frac{8 cd\tilde{Y}d_{\phi}}{3} + \frac{4 d\tilde{X}\tilde{Y}d_{\phi}\cos\theta}{3} + \frac{20 c\tilde{Y}^2d_{\phi}\cos\theta}{3} - \frac{9 \tilde{X}\tilde{I}_{u}^2d_{\phi}}{2} \\
 &\quad - \frac{5 \tilde{X}^2d_{\phi}^2}{6} + \frac{11 \tilde{Y}^2d_{\phi}^2}{6} + 2 \tilde{X}D_{\phi} - 8 \tilde{X}\tilde{Y}N_{\phi}\sin\theta + \tilde{X}^2P_{\phi}\sin\theta \\
 &\quad - 7 \tilde{Y}^2P_{\phi}\sin\theta + 4 P\tilde{X}\tilde{X}_{\phi}\sin\theta + \frac{10 c^2\tilde{Y}\tilde{X}_{\phi}\cos\theta}{3} \\
 &\quad + \frac{10 d^2\tilde{Y}\tilde{X}_{\phi}\cos\theta}{3} - 2 N\tilde{Y}\tilde{X}_{\phi}\sin\theta - 9 \tilde{Y}\tilde{I}_{u}^2\tilde{X}_{\phi}\cos\theta \\
 &\quad - \frac{40 c\tilde{X}c_{\phi}\tilde{X}_{\phi}}{3} - \frac{8 d\tilde{Y}c_{\phi}\tilde{X}_{\phi}}{3} - \frac{28 d\tilde{X}d_{\phi}\tilde{X}_{\phi}}{3} - \frac{4 c\tilde{Y}d_{\phi}\tilde{X}_{\phi}}{3} \\
 &\quad - 4 c^2\tilde{X}_{\phi}^2 - 2 d^2\tilde{X}_{\phi}^2 + \frac{2 c^2\tilde{X}\tilde{Y}_{\phi}\cos\theta}{3} + \frac{2 d^2\tilde{X}\tilde{Y}_{\phi}\cos\theta}{3} \\
 &\quad - 10 N\tilde{X}\tilde{Y}_{\phi}\sin\theta - 8 P\tilde{Y}\tilde{Y}_{\phi}\sin\theta - 9 \tilde{X}\tilde{I}_{u}^2\tilde{Y}_{\phi}\cos\theta + \frac{8 d\tilde{X}c_{\phi}\tilde{Y}_{\phi}}{3} \\
 &\quad - \frac{10 c\tilde{Y}c_{\phi}\tilde{Y}_{\phi}}{3} - \frac{20 c\tilde{X}d_{\phi}\tilde{Y}_{\phi}}{3} - \frac{22 d\tilde{Y}d_{\phi}\tilde{Y}_{\phi}}{3} - 4 cd\tilde{X}_{\phi}\tilde{Y}_{\phi} \\
 &\quad - \frac{c^2\tilde{Y}_{\phi}^2}{2} - \frac{5 d^2\tilde{Y}_{\phi}^2}{2} - \frac{9 \tilde{I}_{u}^2\tilde{Y}_{\phi}^2}{4} - \frac{11 c\tilde{X}^2c_{\phi\phi}}{6} + \frac{4 d\tilde{X}\tilde{Y}c_{\phi\phi}}{3} \\
 &\quad + \frac{5 c\tilde{Y}^2c_{\phi\phi}}{6} - \frac{11 d\tilde{X}^2d_{\phi\phi}}{6} - \frac{4 c\tilde{X}\tilde{Y}d_{\phi\phi}}{3} + \frac{5 d\tilde{Y}^2d_{\phi\phi}}{6} \\
 &\quad - 8 c^2\tilde{X}\tilde{X}_{\phi\phi} - 4 d^2\tilde{X}\tilde{X}_{\phi\phi} - 4 cd\tilde{Y}\tilde{X}_{\phi\phi} - 4 cd\tilde{X}\tilde{Y}_{\phi\phi} - c^2\tilde{Y}\tilde{Y}_{\phi\phi} \\
 &\quad - 5 d^2\tilde{Y}\tilde{Y}_{\phi\phi} - \frac{9 \tilde{Y}\tilde{I}_{u}^2\tilde{Y}_{\phi\phi}}{2} - \frac{28 c^2\tilde{X}c_{\theta}\sin\theta}{9} - \frac{4 d^2\tilde{X}c_{\theta}\sin\theta}{9} \\
 &\quad + \frac{c\tilde{X}^2c_{\theta}\cos\theta\sin\theta}{2} - \frac{8 cd\tilde{Y}c_{\theta}\sin\theta}{3} - \frac{20 d\tilde{X}\tilde{Y}c_{\theta}\cos\theta\sin\theta}{3} \\
 &\quad - \frac{7 c\tilde{Y}^2c_{\theta}\cos\theta\sin\theta}{2} - \frac{9 \tilde{X}\tilde{I}_{u}^2c_{\theta}\sin\theta}{2} + \frac{16 \tilde{X}\tilde{Y}c_{\phi}c_{\theta}\sin\theta}{3} \\
 &\quad + 4 \tilde{X}^2d_{\phi}c_{\theta}\sin\theta + 4 \tilde{Y}^2d_{\phi}c_{\theta}\sin\theta + \frac{20 d\tilde{X}\tilde{X}_{\phi}c_{\theta}\sin\theta}{3} \\
 &\quad + \frac{5 c\tilde{Y}\tilde{X}_{\phi}c_{\theta}\sin\theta}{3} + \frac{c\tilde{X}\tilde{Y}_{\phi}c_{\theta}\sin\theta}{3} - \frac{4 d\tilde{Y}\tilde{Y}_{\phi}c_{\theta}\sin\theta}{3} \\
 &\quad + \frac{11 \tilde{X}^2c_{\theta}^2\sin^2\theta}{6} - \frac{5 \tilde{Y}^2c_{\theta}^2\sin^2\theta}{6} + 2 \tilde{X}C_{\theta}\sin\theta \\
 &\quad - \frac{8 cd\tilde{X}d_{\theta}\sin\theta}{3} + \frac{d\tilde{X}^2d_{\theta}\cos\theta\sin\theta}{2} - \frac{4 c^2\tilde{Y}d_{\theta}\sin\theta}{9} \\
 &\quad - \frac{28 d^2\tilde{Y}d_{\theta}\sin\theta}{9} + \frac{10 c\tilde{X}\tilde{Y}d_{\theta}\sin 2\theta}{3} - \frac{7 d\tilde{Y}^2d_{\theta}\cos\theta\sin\theta}{2} \\
 &\quad - \frac{9 \tilde{Y}\tilde{I}_{u}^2d_{\theta}\sin\theta}{2} - 4 \tilde{X}^2c_{\phi}d_{\theta}\sin\theta - 4 \tilde{Y}^2c_{\phi}d_{\theta}\sin\theta \\
 &\quad + \frac{16 \tilde{X}\tilde{Y}d_{\phi}d_{\theta}\sin\theta}{3} - \frac{8 c\tilde{X}\tilde{X}_{\phi}d_{\theta}\sin\theta}{3} + \frac{17 d\tilde{Y}\tilde{X}_{\phi}d_{\theta}\sin\theta}{3} \\
 &\quad + \frac{13 d\tilde{X}\tilde{Y}_{\phi}d_{\theta}\sin\theta}{3} - \frac{8 c\tilde{Y}\tilde{Y}_{\phi}d_{\theta}\sin\theta}{3} + \frac{11 \tilde{X}^2d_{\theta}^2\sin^2\theta}{6} \\
 &\quad - \frac{5 \tilde{Y}^2d_{\theta}^2\sin^2\theta}{6} + 2 \tilde{Y}D_{\theta}\sin\theta - 7 \tilde{X}^2N_{\theta}\sin^2\theta + \tilde{Y}^2N_{\theta}\sin^2\theta \\
 &\quad - 8 \tilde{X}\tilde{Y}P_{\theta}\sin^2\theta - \frac{23 c^2\tilde{X}\tilde{X}_{\theta}\cos\theta\sin\theta}{3} - \frac{35 d^2\tilde{X}\tilde{X}_{\theta}\cos\theta\sin\theta}{3} \\
 &\quad - 8 N\tilde{X}\tilde{X}_{\theta}\sin^2\theta + 4 cd\tilde{Y}\tilde{X}_{\theta}\cos\theta\sin\theta - 10 P\tilde{Y}\tilde{X}_{\theta}\sin^2\theta \\
 &\quad - \frac{45 \tilde{X}\tilde{I}_{u}^2\tilde{X}_{\theta}\cos\theta\sin\theta}{2} + \frac{4 d\tilde{X}c_{\phi}\tilde{X}_{\theta}\sin\theta}{3} + \frac{c\tilde{Y}c_{\phi}\tilde{X}_{\theta}\sin\theta}{3} \\
 &\quad + \frac{8 c\tilde{X}d_{\phi}\tilde{X}_{\theta}\sin\theta}{3} + \frac{13 d\tilde{Y}d_{\phi}\tilde{X}_{\theta}\sin\theta}{3} + 4 cd\tilde{X}_{\phi}\tilde{X}_{\theta}\sin\theta \\
 &\quad - c^2\tilde{Y}_{\phi}\tilde{X}_{\theta}\sin\theta + 3 d^2\tilde{Y}_{\phi}\tilde{X}_{\theta}\sin\theta - \frac{9 \tilde{I}_{u}^2\tilde{Y}_{\phi}\tilde{X}_{\theta}\sin\theta}{2} \\
 &\quad - \frac{5 c\tilde{X}c_{\theta}\tilde{X}_{\theta}}{3} + \frac{5 c\tilde{X}c_{\theta}\tilde{X}_{\theta}\cos 2\theta}{3} - \frac{8 d\tilde{Y}c_{\theta}\tilde{X}_{\theta}\sin^2\theta}{3} \\
 &\quad - \frac{22 d\tilde{X}d_{\theta}\tilde{X}_{\theta}\sin^2\theta}{3} + \frac{20 c\tilde{Y}d_{\theta}\tilde{X}_{\theta}\sin^2\theta}{3} - \frac{c^2\tilde{X}_{\theta}^2\sin^2\theta}{2} \\
 &\quad - \frac{5 d^2\tilde{X}_{\theta}^2\sin^2\theta}{2} - \frac{9 \tilde{I}_{u}^2\tilde{X}_{\theta}^2\sin^2\theta}{4} + 4 cd\tilde{X}\tilde{Y}_{\theta}\cos\theta\sin\theta \\
 &\quad - 2 P\tilde{X}\tilde{Y}_{\theta}\sin^2\theta - \frac{56 c^2\tilde{Y}\tilde{Y}_{\theta}\cos\theta\sin\theta}{3} - \frac{44 d^2\tilde{Y}\tilde{Y}_{\theta}\cos\theta\sin\theta}{3} \\
 &\quad + 4 N\tilde{Y}\tilde{Y}_{\theta}\sin^2\theta + \frac{5 c\tilde{X}c_{\phi}\tilde{Y}_{\theta}\sin\theta}{3} - \frac{20 d\tilde{Y}c_{\phi}\tilde{Y}_{\theta}\sin\theta}{3} \\
 &\quad + \frac{17 d\tilde{X}d_{\phi}\tilde{Y}_{\theta}\sin\theta}{3} + \frac{8 c\tilde{Y}d_{\phi}\tilde{Y}_{\theta}\sin\theta}{3} + 4 d^2\tilde{X}_{\phi}\tilde{Y}_{\theta}\sin\theta \\
 &\quad - 4 cd\tilde{Y}_{\phi}\tilde{Y}_{\theta}\sin\theta + \frac{8 d\tilde{X}c_{\theta}\tilde{Y}_{\theta}\sin^2\theta}{3} - \frac{40 c\tilde{Y}c_{\theta}\tilde{Y}_{\theta}\sin^2\theta}{3} \\
 &\quad + \frac{2 c\tilde{X}d_{\theta}\tilde{Y}_{\theta}}{3} - \frac{2 c\tilde{X}d_{\theta}\tilde{Y}_{\theta}\cos 2\theta}{3} - \frac{28 d\tilde{Y}d_{\theta}\tilde{Y}_{\theta}\sin^2\theta}{3} \\
 &\quad + 4 cd\tilde{X}_{\theta}\tilde{Y}_{\theta}\sin^2\theta - 4 c^2\tilde{Y}_{\theta}^2\sin^2\theta - 2 d^2\tilde{Y}_{\theta}^2\sin^2\theta + \frac{4 d\tilde{X}^2c_{\theta\phi}\sin\theta}{3} \\
 &\quad + \frac{16 c\tilde{X}\tilde{Y}c_{\theta\phi}\sin\theta}{3} - \frac{4 d\tilde{Y}^2c_{\theta\phi}\sin\theta}{3} - \frac{4 c\tilde{X}^2d_{\theta\phi}\sin\theta}{3} \\
 &\quad + \frac{16 d\tilde{X}\tilde{Y}d_{\theta\phi}\sin\theta}{3} + \frac{4 c\tilde{Y}^2d_{\theta\phi}\sin\theta}{3} + 8 cd\tilde{X}\tilde{X}_{\theta\phi}\sin\theta \\
 &\quad - c^2\tilde{Y}\tilde{X}_{\theta\phi}\sin\theta + 7 d^2\tilde{Y}\tilde{X}_{\theta\phi}\sin\theta - \frac{9 \tilde{Y}\tilde{I}_{u}^2\tilde{X}_{\theta\phi}\sin\theta}{2} \\
 &\quad - c^2\tilde{X}\tilde{Y}_{\theta\phi}\sin\theta + 7 d^2\tilde{X}\tilde{Y}_{\theta\phi}\sin\theta - 8 cd\tilde{Y}\tilde{Y}_{\theta\phi}\sin\theta \\
 &\quad - \frac{9 \tilde{X}\tilde{I}_{u}^2\tilde{Y}_{\theta\phi}\sin\theta}{2} + \frac{5 c\tilde{X}^2c_{\theta\theta}\sin^2\theta}{6} - \frac{4 d\tilde{X}\tilde{Y}c_{\theta\theta}\sin^2\theta}{3} \\
 &\quad - \frac{11 c\tilde{Y}^2c_{\theta\theta}\sin^2\theta}{6} + \frac{5 d\tilde{X}^2d_{\theta\theta}\sin^2\theta}{6} + \frac{2 c\tilde{X}\tilde{Y}d_{\theta\theta}}{3} \\
 &\quad - \frac{2 c\tilde{X}\tilde{Y}d_{\theta\theta}\cos 2\theta}{3} - \frac{11 d\tilde{Y}^2d_{\theta\theta}\sin^2\theta}{6} - c^2\tilde{X}\tilde{X}_{\theta\theta}\sin^2\theta \\
 &\quad - 5 d^2\tilde{X}\tilde{X}_{\theta\theta}\sin^2\theta + 4 cd\tilde{Y}\tilde{X}_{\theta\theta}\sin^2\theta - \frac{9 \tilde{X}\tilde{I}_{u}^2\tilde{X}_{\theta\theta}\sin^2\theta}{2} \\
 &\quad + 4 cd\tilde{X}\tilde{Y}_{\theta\theta}\sin^2\theta - 8 c^2\tilde{Y}\tilde{Y}_{\theta\theta}\sin^2\theta - 4 d^2\tilde{Y}\tilde{Y}_{\theta\theta}\sin^2\theta \\
 &\quad + 2 c^2\tilde{X}^2c_{u}\sin^2\theta - 2 d^2\tilde{X}^2c_{u}\sin^2\theta + 8 cd\tilde{X}\tilde{Y}c_{u}\sin^2\theta \\
 &\quad - 2 c^2\tilde{Y}^2c_{u}\sin^2\theta + 2 d^2\tilde{Y}^2c_{u}\sin^2\theta - 4 \tilde{X}^2C_{u}\sin^2\theta \\
 &\quad + 4 \tilde{Y}^2C_{u}\sin^2\theta + 4 cd\tilde{X}^2d_{u}\sin^2\theta - 4 c^2\tilde{X}\tilde{Y}d_{u}\sin^2\theta \\
 &\quad + 4 d^2\tilde{X}\tilde{Y}d_{u}\sin^2\theta - 4 cd\tilde{Y}^2d_{u}\sin^2\theta - 8 \tilde{X}\tilde{Y}D_{u}\sin^2\theta, \\
           \tensor[_{\Lambda^3}]{\overset{2}{\Omega}}{_{00}} &=   
 - \frac{4 c^4\tilde{X}^2\sin^2\theta}{3} - \frac{4 cC\tilde{X}^2\sin^2\theta}{3} - \frac{8 c^2d^2\tilde{X}^2\sin^2\theta}{3} \\
 &\quad - \frac{4 d^4\tilde{X}^2\sin^2\theta}{3} - \frac{4 dD\tilde{X}^2\sin^2\theta}{3} + \frac{4 \overset{4}{\gamma}\tilde{X}^2\sin^2\theta}{3} \\
 &\quad + 2 cd^2\tilde{X}^3\cos\theta\sin^2\theta + c^3\tilde{X}^3\sin\theta\sin 2\theta - 6 C\tilde{X}^3\sin\theta\sin 2\theta \\
 &\quad + \frac{8 \overset{4}{\delta}\tilde{X}\tilde{Y}\sin^2\theta}{3} + 4 c^2d\tilde{X}^2\tilde{Y}\cos\theta\sin^2\theta - 24 D\tilde{X}^2\tilde{Y}\cos\theta\sin^2\theta \\
 &\quad + 2 d^3\tilde{X}^2\tilde{Y}\sin\theta\sin 2\theta - \frac{4 c^4\tilde{Y}^2\sin^2\theta}{3} - \frac{4 cC\tilde{Y}^2\sin^2\theta}{3} \\
 &\quad - \frac{8 c^2d^2\tilde{Y}^2\sin^2\theta}{3} - \frac{4 d^4\tilde{Y}^2\sin^2\theta}{3} - \frac{4 dD\tilde{Y}^2\sin^2\theta}{3} \\
 &\quad - \frac{4 \overset{4}{\gamma}\tilde{Y}^2\sin^2\theta}{3} - 2 cd^2\tilde{X}\tilde{Y}^2\cos\theta\sin^2\theta - c^3\tilde{X}\tilde{Y}^2\sin\theta\sin 2\theta \\
 &\quad + 6 C\tilde{X}\tilde{Y}^2\sin\theta\sin 2\theta - 3 c^2\tilde{X}^2\tilde{I}_{u}^2\sin^2\theta - 3 d^2\tilde{X}^2\tilde{I}_{u}^2\sin^2\theta \\
 &\quad - 3 c^2\tilde{Y}^2\tilde{I}_{u}^2\sin^2\theta - 3 d^2\tilde{Y}^2\tilde{I}_{u}^2\sin^2\theta + 2 c^2\tilde{X}^2\tilde{Y}c_{\phi}\sin^2\theta \\
 &\quad - 2 d^2\tilde{X}^2\tilde{Y}c_{\phi}\sin^2\theta + 8 cd\tilde{X}\tilde{Y}^2c_{\phi}\sin^2\theta - 2 c^2\tilde{Y}^3c_{\phi}\sin^2\theta \\
 &\quad + 2 d^2\tilde{Y}^3c_{\phi}\sin^2\theta - 4 \tilde{X}^2\tilde{Y}C_{\phi}\sin^2\theta + 4 \tilde{Y}^3C_{\phi}\sin^2\theta \\
 &\quad + 4 cd\tilde{X}^2\tilde{Y}d_{\phi}\sin^2\theta - 4 c^2\tilde{X}\tilde{Y}^2d_{\phi}\sin^2\theta + 4 d^2\tilde{X}\tilde{Y}^2d_{\phi}\sin^2\theta \\
 &\quad - 4 cd\tilde{Y}^3d_{\phi}\sin^2\theta - 8 \tilde{X}\tilde{Y}^2D_{\phi}\sin^2\theta + 4 c^2d\tilde{X}^2\tilde{X}_{\phi}\sin^2\theta \\
 &\quad + 4 d^3\tilde{X}^2\tilde{X}_{\phi}\sin^2\theta + 6 D\tilde{X}^2\tilde{X}_{\phi}\sin^2\theta - 4 c^3\tilde{X}\tilde{Y}\tilde{X}_{\phi}\sin^2\theta \\
 &\quad - 8 C\tilde{X}\tilde{Y}\tilde{X}_{\phi}\sin^2\theta - 4 cd^2\tilde{X}\tilde{Y}\tilde{X}_{\phi}\sin^2\theta - 2 D\tilde{Y}^2\tilde{X}_{\phi}\sin^2\theta \\
 &\quad + \frac{9 d\tilde{X}^2\tilde{I}_{u}^2\tilde{X}_{\phi}\sin^2\theta}{2} + \frac{9 d\tilde{Y}^2\tilde{I}_{u}^2\tilde{X}_{\phi}\sin^2\theta}{2} - c^3\tilde{X}^2\tilde{Y}_{\phi}\sin^2\theta \\
 &\quad - 8 C\tilde{X}^2\tilde{Y}_{\phi}\sin^2\theta - cd^2\tilde{X}^2\tilde{Y}_{\phi}\sin^2\theta + 2 c^2d\tilde{X}\tilde{Y}\tilde{Y}_{\phi}\sin^2\theta \\
 &\quad + 2 d^3\tilde{X}\tilde{Y}\tilde{Y}_{\phi}\sin^2\theta - 12 D\tilde{X}\tilde{Y}\tilde{Y}_{\phi}\sin^2\theta - 3 c^3\tilde{Y}^2\tilde{Y}_{\phi}\sin^2\theta \\
 &\quad + 4 C\tilde{Y}^2\tilde{Y}_{\phi}\sin^2\theta - 3 cd^2\tilde{Y}^2\tilde{Y}_{\phi}\sin^2\theta - \frac{9 c\tilde{X}^2\tilde{I}_{u}^2\tilde{Y}_{\phi}\sin^2\theta}{2} \\
 &\quad - \frac{9 c\tilde{Y}^2\tilde{I}_{u}^2\tilde{Y}_{\phi}\sin^2\theta}{2} + 2 c^2\tilde{X}^3c_{\theta}\sin^3\theta - 2 d^2\tilde{X}^3c_{\theta}\sin^3\theta \\
 &\quad + 8 cd\tilde{X}^2\tilde{Y}c_{\theta}\sin^3\theta - 2 c^2\tilde{X}\tilde{Y}^2c_{\theta}\sin^3\theta + 2 d^2\tilde{X}\tilde{Y}^2c_{\theta}\sin^3\theta \\
 &\quad - 4 \tilde{X}^3C_{\theta}\sin^3\theta + 4 \tilde{X}\tilde{Y}^2C_{\theta}\sin^3\theta + 4 cd\tilde{X}^3d_{\theta}\sin^3\theta \\
 &\quad - 4 c^2\tilde{X}^2\tilde{Y}d_{\theta}\sin^3\theta + 4 d^2\tilde{X}^2\tilde{Y}d_{\theta}\sin^3\theta - 4 cd\tilde{X}\tilde{Y}^2d_{\theta}\sin^3\theta \\
 &\quad - 8 \tilde{X}^2\tilde{Y}D_{\theta}\sin^3\theta + 3 c^3\tilde{X}^2\tilde{X}_{\theta}\sin^3\theta - 4 C\tilde{X}^2\tilde{X}_{\theta}\sin^3\theta \\
 &\quad + 3 cd^2\tilde{X}^2\tilde{X}_{\theta}\sin^3\theta + 2 c^2d\tilde{X}\tilde{Y}\tilde{X}_{\theta}\sin^3\theta + 2 d^3\tilde{X}\tilde{Y}\tilde{X}_{\theta}\sin^3\theta \\
 &\quad - 12 D\tilde{X}\tilde{Y}\tilde{X}_{\theta}\sin^3\theta + c^3\tilde{Y}^2\tilde{X}_{\theta}\sin^3\theta + 8 C\tilde{Y}^2\tilde{X}_{\theta}\sin^3\theta \\
 &\quad + cd^2\tilde{Y}^2\tilde{X}_{\theta}\sin^3\theta + \frac{9 c\tilde{X}^2\tilde{I}_{u}^2\tilde{X}_{\theta}\sin^3\theta}{2} + \frac{9 c\tilde{Y}^2\tilde{I}_{u}^2\tilde{X}_{\theta}\sin^3\theta}{2} \\
 &\quad - 2 D\tilde{X}^2\tilde{Y}_{\theta}\sin^3\theta + 4 c^3\tilde{X}\tilde{Y}\tilde{Y}_{\theta}\sin^3\theta + 8 C\tilde{X}\tilde{Y}\tilde{Y}_{\theta}\sin^3\theta \\
 &\quad + 4 cd^2\tilde{X}\tilde{Y}\tilde{Y}_{\theta}\sin^3\theta + 4 c^2d\tilde{Y}^2\tilde{Y}_{\theta}\sin^3\theta + 4 d^3\tilde{Y}^2\tilde{Y}_{\theta}\sin^3\theta \\
 &\quad + 6 D\tilde{Y}^2\tilde{Y}_{\theta}\sin^3\theta + \frac{9 d\tilde{X}^2\tilde{I}_{u}^2\tilde{Y}_{\theta}\sin^3\theta}{2} + \frac{9 d\tilde{Y}^2\tilde{I}_{u}^2\tilde{Y}_{\theta}\sin^3\theta}{2}.            
\end{align*}

\bigskip

\footnotesize {

\noindent {\bf Acknowledgement} The work is supported by the National Natural Science Foundation of China 12326602.


}


\end{document}